\documentclass[a4paper,autoref]{lipics-v2021}
\usepackage{amsmath,amssymb,amsfonts,latexsym}
\usepackage{color,graphicx}
\usepackage{url} 
\usepackage{fancybox}
\usepackage{algorithm}
\usepackage[noend]{algpseudocode}
\usepackage{mathtools}
\usepackage{enumerate,xspace}
\usepackage{thm-restate}
\usepackage{siunitx}

\nolinenumbers
\hideLIPIcs

\newcommand{\mypara}[1]{\medskip\noindent{\sf\textbf{#1}}}  

\newcommand{\eps}{\varepsilon}
\renewcommand{\epsilon}{\eps}
\newcommand{\etal}{\emph{et al.}\xspace}

\theoremstyle{plain}

\newenvironment{myquote}%
  {\list{}{\leftmargin=4mm\rightmargin=4mm}\item[]}%
  {\endlist}
\newenvironment{claiminproof}{\begin{myquote}\noindent\emph{Claim.}}{\end{myquote}}
\newenvironment{proofinproof}{\begin{myquote}\noindent\emph{Proof.}}{\hfill $\lhd$ \end{myquote}}
\newenvironment{obsinproof}{\begin{myquote}\noindent\emph{Observation.}}{\end{myquote}}
\newcommand{\anc}{\mathrm{anc}}

\newcommand{\B}{\ensuremath{\mathcal{B}}}
\newcommand{\C}{\ensuremath{\mathcal{C}}}

\newcommand{\G}{\ensuremath{\mathcal{G}}}

\newcommand{\R}{\ensuremath{\mathcal{R}}}

\newcommand{\REAL}{\ensuremath{\mathbb{R}}}
\newcommand{\Reals}{\REAL}

\renewcommand{\leq}{\leqslant}
\renewcommand{\geq}{\geqslant}

\newcommand{\np}{{\sc np}\xspace}

\DeclareMathOperator{\cost}{cost}
\newcommand{\costa}{\cost_{\alpha}}
\newcommand{\opt}{\mbox{{\sc opt}}\xspace}
\newcommand{\optsub}{\mathrm{opt}}
\newcommand{\alg}{\mbox{{\sc alg}}\xspace}
\newcommand{\range}{\rho}
\newcommand{\tree}{\mathcal{T}}
\DeclareMathOperator{\br}{best-range}
\DeclareMathOperator{\mc}{min-cost}
\DeclareMathOperator{\myroot}{root}
\DeclareMathOperator{\prd}{pred}
\DeclareMathOperator{\suc}{succ}
\newcommand{\nil}{\mbox{\sc nil}\xspace}
\newcommand{\graph}{\G}
\newcommand{\cg}{\graph_{\range}}
\newcommand{\Zleft}{Z_{\mathrm{left}}}
\newcommand{\Zright}{Z_{\mathrm{right}}}
\newcommand{\Pleft}{P_{\mathrm{left}}}
\newcommand{\Pright}{P_{\mathrm{right}}}
\newcommand{\Pmid}{P_{\mathrm{mid}}}
\newcommand{\old}{\mathrm{old}}
\newcommand{\rold}{\range_{\old}}

\newcommand{\mynew}{\mathrm{new}}
\newcommand{\rnew}{\range_{\mynew}}

\newcommand{\ropt}{\range_{\optsub}}
\newcommand{\standard}{\mathrm{st}}
\newcommand{\rst}{\range_{\standard}}
\newcommand{\rk}{\range_k}
\newcommand{\ralg}{\range_{\mathrm{alg}}}
\newcommand{\rsb}{\range_{\mathrm{sb}}}
\newcommand{\rmst}{\range_{\mathrm{mst}}}
\newcommand{\Csb}{C_{\mathrm{sb}}}
\newcommand{\Copt}{C_{\mathrm{\optsub}}}
\newcommand{\Pexp}{P_{\mathrm{exp}}}
\newcommand{\Pcheap}{P_{\mathrm{cheap}}}
\newcommand{\dmax}{d_{\mathrm{max}}}

\DeclareMathOperator{\degree}{deg}

\newcommand{\cdelta}{c_{\delta}}

\newcommand{\cw}{\mathrm{cw}}           
\newcommand{\ccw}{\mathrm{ccw}}         
\newcommand{\dcw}{d_{\cw}}              
\newcommand{\dccw}{d_{\ccw}}            
\newcommand{\dhop}{d_{\mathrm{hop}}}      
\newcommand{\bbS}{\mathbb{S}}    
\newcommand{\cov}{\mathrm{cov}}         
\newcommand{\pa}{\mathrm{pa}}           
\newcommand{\CW}{{\sc cw}}              
\newcommand{\CCW}{{\sc ccw}}            
\newcommand{\exc}{\mathrm{excess}}
\newcommand{\ca}{c_{\alpha}}

\title{Stable Approximation Algorithms for the \\ {} Dynamic Broadcast Range-Assignment Problem }
\funding{MdB, AS, and FS are supported by the  Dutch Research Council (NWO) through    Gravitation-grant NETWORKS-024.002.003.}
\author{Mark de Berg}{Department of Mathematics and Computer Science, TU Eindhoven, the Netherlands}{M.T.d.Berg@tue.nl}{}{}
\author{Arpan Sadhukhan}{Department of Mathematics and Computer Science, TU Eindhoven, the Netherlands}{A.Sadhukhan@tue.nl}{}{}
\author{Frits Spieksma}{Department of Mathematics and Computer Science, TU Eindhoven, the Netherlands}{f.c.r.spieksma@tue.nl}{}{}
\authorrunning{M.~de Berg, A.~Sadhukhan and F.~Spieksma} 
\Copyright{Mark de Berg and Arpan Sadhukhan and Frits Spieksma}

\acknowledgements{We thank the reviewers of an earlier version of the paper for pointing us to some important
references and for other helpful comments.}



\begin{document}
\setcounter{page}{0}
\maketitle

\begin{abstract}
Let $P$ be a set of points in $\Reals^d$ (or some other metric space), where each point $p\in P$ has an associated
transmission range, denoted~$\range(p)$. The range assignment~$\range$ induces a directed
communication graph $\cg(P)$ on~$P$, which contains an edge $(p,q)$ iff $|pq| \leq \range(p)$.
In the broadcast range-assignment problem, the goal is to assign the ranges such that $\cg(P)$
contains an arborescence rooted at a designated root node and the cost $\sum_{p \in P} \range(p)^2$
of the assignment is minimized.
 
We study the dynamic version of this problem. In particular, we study trade-offs 
between the stability of the solution---the number of ranges that are modified when a point
is inserted into or deleted from $P$---and its approximation ratio. To this end we introduce the
concept of \emph{$k$-stable algorithms}, which are algorithms that modify the range 
of at most $k$ points when they update the solution. We also introduce the concept of a 
\emph{stable approximation scheme}, or \emph{SAS} for short.
A SAS is an update algorithm $\alg$ that,
for any given fixed parameter $\eps>0$, is $k(\eps)$-stable and that maintains a solution with 
approximation ratio~$1+\eps$, where the stability parameter $k(\eps)$ only depends on $\eps$ 
and not on the size of~$P$.
%
We study such trade-offs in three settings.
\begin{itemize}
\item For the problem in $\Reals^1$, we present a SAS with $k(\eps)=O(1/\eps)$.
      Furthermore, we prove that this is tight in the worst case:
      any SAS for the problem must have $k(\eps)=\Omega(1/\eps)$.
      We also present algorithms with very small stability parameters:
      a 1-stable $(6+2\sqrt{5})$-approximation algorithm---this
      algorithm can only handle insertions---a (trivial) $2$-stable 2-approximation algorithm, 
      and a $3$-stable $1.97$-approximation algorithm. 
\item For the problem in $\bbS^1$ (that is, when the underlying space is a circle) we prove that no SAS exists. This is in spite of the fact that, for the static problem in $\bbS^1$, we prove that an optimal solution can always be obtained by cutting the circle at an appropriate point and solving the resulting problem in~$\Reals^1$.
\item For the problem in $\Reals^2$, we also prove that no SAS exists, and we present a $O(1)$-stable $O(1)$-approximation algorithm. 
\end{itemize}
Most results generalize to when the range-assignment cost is 
$\sum_{p \in P} \range(p)^{\alpha}$, for some constant~$\alpha > 1$.

\end{abstract}

\section{Introduction}

\mypara{The broadcast range-assignment problem.}
Let $P$ be a set of points in $\Reals^d$, representing transmission devices
in a wireless network. By assigning each point $p\in P$ a transmission range~$\range(p)$, we
obtain a \emph{communication graph}~$\cg(P)$. The nodes in $\cg(P)$ are the points
from~$P$ and there is a directed edge~$(p,q)$ iff $|pq|\leq \range(p)$, where
$|pq|$ denotes the Euclidean distance between $p$ and $q$. The energy consumption
of a device depends on its transmission range: the larger the range, the more energy it needs.
More precisely, the energy needed to obtain a transmission range~$\rho(p)$ is given by
$\range(p)^{\alpha}$, for some real constant $\alpha > 1$ called the \emph{distance-power gradient}.
In practice, $\alpha$ depends on the environment and ranges from~1 to~6~\cite{Pahlavan_Levesque_book}.
Thus the overall cost of a range assignment is $\cost_{\alpha}(\range(P)):=\sum_{p \in P} \range(p)^ \alpha$,
where 
we use $\range(P)$ to denote the set of ranges
given to the points in $P$ by the assignment~$\range$.
The goal of the range-assignment problem is to assign the ranges such that 
$\cg(P)$ has certain connectivity properties while minimizing the total cost~\cite{Survey_Clementi}.
Desirable connectivity properties are that $\cg(P)$ is ($h$-hop) strongly 
connected~\cite{cpfps-mrap-03,cps-hrpra-99,DBLP:conf/stacs/ClementiPS00,kkkp-pcprn-00}
or that $\cg(P)$ contains a \emph{broadcast tree}, that is, an arborescence rooted at
a given source~$s\in P$. The latter property leads to the \emph{broadcast range-assignment
problem}, which is the topic of our paper. 

The broadcast range-assignment problem has been studied extensively, sometimes with the extra 
condition that any point in $P$ is reachable in at most $h$ hops from the source~$s$.
For $\alpha=1$ the problem is trivial in any dimension: setting the range of the source~$s$ to $\max\{|sp|: p\in P\}$ and all other ranges to zero is optimal; however, for
any $\alpha>1$ the problem is \np-hard in $\Reals^d$ for $d\geq 2$~\cite{Clementi_range_assignment_complexity, Fuchs_range_assigment_hardness}.
Approximation algorithms and results on hardness of approximation are known as 
well~\cite{Clementi_range_assignment_1D,Fuchs_range_assigment_hardness,DBLP:conf/isaac/CaragiannisKK02}.
Many of our results will be on the 1-dimensional (or: linear) broadcast range-assignment problem.
Linear networks are important for modeling road traffic information
systems~\cite{DBLP:conf/sac/BassiouniF99,DBLP:journals/winet/MatharM96} and as such they have 
received ample attention. In $\Reals^1$, the broadcast range-assignment problem
is no longer \np-hard, and several polynomial-time algorithms have been proposed,
for the standard version, the $h$-hop version, as well as the weighted
version~\cite{Clementi_range_assignment_1D,DBLP:conf/isaac/CaragiannisKK02,DBLP:journals/tcs/DasDN06,DBLP:journals/ipl/DasN08,DBLP:journals/tvt/AtaeiBK12}.
The currently fastest algorithms for the (standard and $h$-hop) broadcast range-assignment problem
run in $O(n^2)$ time~\cite{DBLP:journals/tcs/DasDN06}.
\medskip

All results mentioned so far are for the static version of the problem. Our interest lies
in the dynamic version, where points can be inserted into and deleted from $P$
(except the source, which should always be present).
This corresponds to new sensors being deployed and existing sensors being removed,
or, in a traffic scenario, cars entering and exiting the highway. 
Recomputing the range assignment from scratch when~$P$
is updated may result in all ranges being changed. The question
we want to answer is therefore: is it possible to maintain a close-to-optimal range assignment
that is relatively stable, that is, an assignment for which only few ranges are modified 
when a point is inserted into or deleted from~$P$? And which trade-offs can be achieved
between the quality of the solution and its stability?

To the best of our knowledge, the dynamic problem has not been studied so far. 
The online problem, where the points from $P$ arrive one by one (there are no deletions)
and it is not allowed to decrease ranges, is studied by De~Berg~\etal~\cite{DBLP:conf/isaac/Berg0U20}.
This restriction is arguably unnatural, and it has the consequence 
that a bounded approximation ratio cannot be achieved. 
Indeed, let the source~$s$ be at $x=0$, and suppose that first
the point $x=1$ arrives, forcing us to set $\range(s) := 1$, and then the
points $x=i/n$ arrive for $1\leq i<n$. In the optimal static solution at the end of this scenario
all points, except the rightmost one, have range~$1/n$; for $\alpha=2$
this induces a total cost of $n \cdot (1/n)^2 = 1/n$. But if we are not allowed to decrease
the range of $s$ after setting $\rho(s)=1$, the total cost will be (at least)~1,
leading to an unbounded approximation ratio. 
Therefore, \cite{DBLP:conf/isaac/Berg0U20} analyze the competitive ratio:
they compare the cost of their algorithm to the cost of an optimal offline algorithm 
(which knows the future arrivals, but must still maintain a valid solution 
at all times without decreasing any range).
As we will see, by allowing to also decrease a few ranges, we are able to maintain
solutions whose cost is close even to the static optimum.

\mypara{Our contribution.}
Before we state our results, we first define the framework we use to analyze our algorithms.
Let $P$ be a dynamic set of points in $\Reals^d$, which includes a fixed source point~$s$
that cannot be deleted. 

An update algorithm \alg for the dynamic broadcast range-assignment problem is an algorithm
that, given the current solution (the current ranges of the points in the current set~$P$)
and the location of the new point to be inserted into $P$, or the point to be deleted from~$P$,
modifies the range assignment so that the updated solution is a valid broadcast range assignment
for the updated set~$P$. 
We call such an update algorithm \emph{$k$-stable} if it modifies at most $k$ ranges
when a point is inserted into or deleted from~$P$.
Here we define the range of a point currently
not in $P$ to be zero. Thus, if a newly inserted point receives a positive range
it will be counted as receiving a modified range; similarly, if a point with positive range is deleted
then it will be counted as receiving a modified range.
To get a more detailed view of the stability, we sometimes distinguish between the number
of increased ranges and the number of decreased ranges, in the worst case. 
When these numbers are $k^+$ and $k^-$, respectively, we say that \alg is \emph{$(k^+,k^-)$-stable}.
This is especially useful when we separately report on the stability of insertions and 
deletions; often, when insertions are $(k_1,k_2)$-stable then deletions will be $(k_2,k_1)$-stable.

We are not only interested in the stability of our update algorithms, but also in the quality
of the solutions they provide. We measure this in the usual way, by considering the approximation
ratio of the solution. As mentioned, we are interested in trade-offs between the
stability of an algorithm and its approximation ratio. Of particular interest are so-called
stable approximation schemes, defined as follows.

\begin{definition}
\label{def:SAS}
A \emph{stable approximation scheme}, or \emph{SAS} for short, is an update algorithm $\alg$ that,
for any given yet fixed parameter $\eps>0$, is $k(\eps)$-stable and that maintains a solution with 
approximation ratio~$1+\eps$, where the stability parameter $k(\eps)$ only depends on $\eps$ 
and not on the size of~$P$.
\end{definition}

\noindent Notice that in the definition of a SAS 
we do not take the computational complexity 
of the update algorithm into account. We point out that, in the context of dynamic scheduling problems (where jobs arrive and disappear in an online fashion, and it is allowed to re-assign jobs), a related concept has been introduced under the name {\em robust PTAS}: a polynomial-time algorithm that, for any given parameter $\epsilon >0$, computes a $(1+\epsilon)$-approximation with re-assignment costs only depending on $\epsilon$, see e.g.  \cite{DBLP:conf/esa/SkutellaV10} and \cite{DBLP:journals/mor/SandersSS09}.

We now present our results. Recall that $\cost_{\alpha}(\range(P)):=\sum_{p \in P} \range(p)^ \alpha$,
is the cost of a range assignment~$\rho$, where ~$\alpha>1$ is a constant.
To make the results easier to interpret, we state the results for~$\alpha=2$;
the dependencies of the bounds on the parameter~$\alpha$ can be found in the theorems presented in later sections.
\begin{itemize}
\item In Section~\ref{sec:sas} we present a SAS for the broadcast range-assignment problem in $\Reals^1$,
      with $k(\eps)=O(1/\eps)$. We prove that this is tight in the worst case,
      by showing that any SAS for the problem must have $k(\eps)=\Omega(1/\eps)$.
\item Our SAS (as well as some other algorithms) needs to know an optimal solution after each update. 
      The fastest existing algorithms to compute an optimal solution in~$\Reals^1$ 
      run in $O(n^2)$ time. In Section~\ref{sec:maintain} we show how to recompute an optimal 
      solution in $O(n\log n)$ time after each update, which we believe to be of independent interest. 
      As a result, our SAS also runs in $O(n\log n)$ time per update. 
\item There is a very simple 2-stable 2-approximation algorithm. We show that a 1-stable
      algorithm with bounded approximation ratio does not exist when both insertions and deletions
      must be handled. For the insertion-only case, however, we give 
      a 1-stable $(6+2\sqrt{5})$-approximation algorithm. We have not been able to
      improve upon the approximation ratio~2 with a 2-stable algorithm,
      but we show that with a 3-stable we can get a $1.97$-approximation.
      Due to lack of space, these results are delegated to the appendix.
\item Next we study the problem in~$\bbS^1$, that is, when the underlying 1-dimensional space is circular.
      This version has, as far as we know, not been studied so far. We first prove that
      in~$\bbS^1$ an optimal solution for the static problem can always be obtained by cutting 
      the circle at an appropriate point and solving the resulting problem in~$\Reals^1$.
      This leads to an algorithm to solve the static problem optimally in $O(n^2\log n)$ time.
      We also prove that, in spite of this, a SAS does not exist in~$\bbS^1$. 
\item Finally, we consider the problem in~$\Reals^2$. Based on the no-SAS proof in~$\bbS^1$,
      we show that the 2-dimensional problem does not admit a SAS either. 
      In addition, we present an $17$-stable $12$-approximation algorithm
      for the 2-dimensional version of the problem. 
\end{itemize}
Omitted proofs and results, and some other additional material, are given in the appendix.
\section{Maintaining an optimal solution in $\Reals^1$}
\label{sec:maintain}
Before we can present our stable algorithms for the broadcast range-assignment problem
in~$\Reals^1$, we first introduce some terminology and we discuss the structure of optimal solutions.
We also present an efficient subroutine to maintain an optimal solution.

\subsection{The structure of an optimal solution}
Several papers have characterized the structure of optimal broadcast range assignments in~$\Reals^1$,
in a more or less explicit manner. We use the characterization by Caragiannis~\etal~\cite{DBLP:conf/isaac/CaragiannisKK02},
which is illustrated in Figure~\ref{fig:optimal-structure} and described next.
\begin{figure}
\begin{center}
\includegraphics{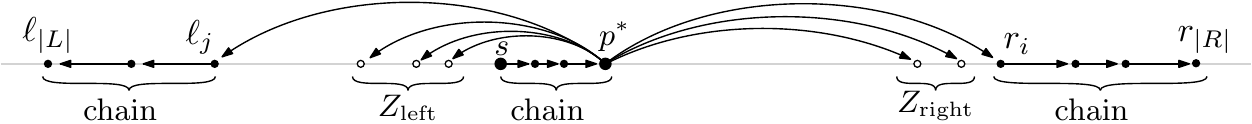}
\end{center}
\caption{The structure of an optimal solution. The non-filled points are zero-range points, 
         the solid black points all have a standard range (for $\ell_{|L|}$ and $r_{|R|}$ 
         the standard range is zero), except for the root-crossing point which (in this example)
         has a long range.}
\label{fig:optimal-structure}
\end{figure}

Let $P := L\cup \{s\}\cup R$ be a point set in $\Reals^1$. Here $s$ is the designated
source node, $L:=\{\ell_1,\ldots,\ell_{|L|}\}$ contains all points from $P$ to the left of~$s$, 
and $R := \{ r_1,\ldots,r_{|R|}\}$ contains all points to the right of~$s$. 
The points in $L$ are numbered in order of increasing distance from~$s$, and the same
is true for the points in $R$. The points $\ell_{|L|}$ and $r_{|R|}$ are called \emph{extreme points}.
In the following, and with a slight abuse of notation, we sometimes
use $p$ or $q$ to refer a generic point from $P$---that is, a point that could be $s$, or 
a point from $R$, or a point from~$L$. Furthermore, we will not distinguish between points
in $P$ and the corresponding nodes in the communication graph~$\cg(P)$.

For a non-extreme point $r_i\in R$, we define $r_{i+1}$ to be its \emph{successor};
similarly, $\ell_{i+1}$ is the successor of~$\ell_i$.
The source~$s$ has (at most) two successors, namely~$r_1$ and~$\ell_1$.
The successor of a point~$p$ is denoted by~$\suc(p)$; for an extreme point~$p$
we define $\suc(p)=\nil$. If $\suc(p)=q\neq\nil$, the
we call $p$ the \emph{predecessor} of $q$ and we write $\prd(q)=p$.
A \emph{chain} is a path in the communication
graph~$\cg(P)$ that only consists of edges connecting a point to its successor.
Thus a chain either visits consecutive points from $\{s\}\cup R$ from left to right,
or it visits consecutive points from $\{s\}\cup L$ from right to left.
It will be convenient to consider the empty path from $s$ to itself to be a chain as well.

Consider a range assignment~$\rho$. We say that a point~$q\in P$ is
\emph{within reach} of a point $p\in P$ if $|pq|\leq \range(p)$.
Let $\B$ a broadcast tree in~$\cg(P)$---that is, $\B$ is an arborescence rooted at~$s$. 
A point in $R\cup L$ in $\B$ is called \emph{root-crossing} in $\B$ if it has 
a child on the other side of~$s$; the source~$s$ is root-crossing if it has a child in $L$ and a child in $R$.
The following theorem, which holds for any distance-power gradient~$\alpha>1$,
is proven in \cite{DBLP:conf/isaac/CaragiannisKK02}.
\begin{theorem}[\cite{DBLP:conf/isaac/CaragiannisKK02}]
\label{thm:structure-1D}
Let $P$ be a point set in~$\Reals^1$. If all points in $P\setminus\{s\}$ lie to 
the same side of the source~$s$, then the optimal solution induces a chain from $s$ to the extreme point in~$P$. 
Otherwise, there is an optimal range assignment~$\range$ such that
$\cg(P)$ contains a broadcast tree~$\B$ with the following structure:
\begin{itemize}
\item $\B$ has a single root-crossing point,~$p^*$.
\item $\B$ contains a chain from $s$ to $p^*$. 
\item All points within reach of $p^*$, except those on the chain from $s$ to $p^*$,
      are children of~$p^*$.
\item Let $r_i$ and $\ell_j$ be the rightmost and leftmost point within reach of~$p^*$, respectively.
      Then $\B$ contains a chain from $r_i$ to~$r_{|R|}$, and  a chain from $\ell_j$ to~$\ell_{|L|}$.
\end{itemize}
\end{theorem}

From now on, whenever we talk about optimal range assignments and their induced broadcast trees, 
we implicitly assume that the broadcast tree has the structure described in Theorem~\ref{thm:structure-1D}. 
Note that the communication graph $\cg(P)$ induced by an optimal range assignment~$\range$
can contain more edges than the ones belonging to the broadcast tree~$\B$. 
Obviously, for~$\range$ to be optimal it must be a minimum-cost assignment
inducing~$\B$. 

Define the \emph{standard range} of a non-extreme point~$r_i\in R$ to be~$|r_i r_{i+1}|$; 
the standard range of the extreme point $r_{|R|}$ is defined to be zero. The standard ranges of the points in~$L$
are defined similarly. 
The source $s$ has two standard ranges, $|s\ell_1|$ and $|sr_1|$.
A range assignment in which every point has a standard range is called a \emph{standard solution};
a standard solution may or may not be optimal.
Note that, in the static problem, it is never useful to give a point a non-zero range that is smaller than its standard range.
Hence, we only need to consider three types of points:
\emph{standard-range points}, \emph{zero-range points}, and \emph{long-range points}.
Here zero-range points are non-extreme points with a zero range, and
a point is said to have a \emph{long range} if its range is greater than its standard range.
Theorem~\ref{thm:structure-1D} implies that an optimal range assignment has the following
properties; see also Figure~\ref{fig:optimal-structure}.
\begin{itemize}   
\item There is at most one long-range point.
\item The set $Z\subset P$ of zero-range points (which may be empty) can be partitioned into two subsets,
      $\Zleft$ and $\Zright$, such that $\Zleft$ consists of consecutive points that lie
      to the left of the source~$s$, and 
      $\Zright$ consists of consecutive points to that lie to the right of~$s$.
\end{itemize}

%


\subsection{An efficient update algorithm}
Using Theorem~\ref{thm:structure-1D} an optimal solution for the broadcast
range-assignment problem can be computed in $O(n^2)$ time~\cite{DBLP:journals/tcs/DasDN06}. 
Below we show that maintaining an optimal solution under insertions and deletions 
can be done more efficiently than by re-computing it from scratch: 
using a suitable data structure, we can update the solution in $O(n\log n)$ time.
This will also be useful in later sections, when we give algorithms that
maintain a stable solution.
\medskip

Recall that an optimal solution for a given point set~$P$ has a single 
root-crossing point,~$p^*$. Once the range~$\range(p^*)$ is fixed, the solution 
is completely determined. Since $\range(p^*)=|p^* p|$ for
some point $p\neq p^*$, there are $n-1$ candidate ranges for a given choice of the 
root-crossing point~$p^*$. The idea of our solution is to implicitly
store the cost of the range assignment for each candidate range of $p^*$ such that, 
upon the insertion or deletion of a point in~$P$,
we can in $O(\log n)$ time find the best range for~$p^*$. By maintaining
$n$~such data structures~$\tree_{p^*}$, one for each choice of the
root-crossing point~$p^*$, we can then find the overall best solution.

\mypara{The data structure for a given root-crossing point.}
Next we explain our data structure for a given candidate root-crossing point~$p^*$.
We assume without loss of generality that $p^*$ lies to the right of
the source point~$s$; it is straightforward to adapt the structure to the 
(symmetric) case where $p^*$ lies to the left of~$s$, and 
to the case where $p^*=s$.

Let $\R_{p^*}$ be the set of all ranges we need to consider for~$p^*$, for the current set~$P$.
The range of a root-crossing point must extend beyond the source point. Hence,
\[
\R_{p^*} := \{ |p^* p| : p \in P \mbox{ and } |p^* p| > |p^* s| \}.
\]
Let $\lambda_1,\ldots, \lambda_m$ denote the sequence of ranges in~$\R_{p^*}$, ordered from small
to large. (If $\R_{p^*}=\emptyset$, there is nothing to do and our
data structure is empty.)  As mentioned, once we fix a range $\lambda_j$ for the given
root-crossing point~$p^*$, the solution is fully determined by Theorem~\ref{thm:structure-1D}:
there is a chain from $s$ to $p^*$, a chain from the rightmost point within range
of $p^*$ to the right-extreme point, and a chain from the leftmost point within range
of $p^*$ to the left-extreme point.
We denote the resulting range assignment\footnote{When $P$ lies completely to one side of~$s$,
then the range assignment is formally not root-crossing. We permit ourselves this slight 
abuse of terminology because by considering~$s$ as root-crossing point, setting 
$\range(s) := |s\suc(s)|$ and adding a chain from $\suc(s)$ to the extreme point,
we get an optimal solution.}
for $P$ by~$\Gamma(P,p^*,\lambda_j)$.

Our data structure, which implicitly stores the costs of the range assignments 
$\Gamma(P,p^*,\lambda_j)$ for all $\lambda_j\in\R_{p^*}$, is an augmented balanced binary 
search tree~$\tree_{p^*}$.
The key to the efficient maintenance of~$\tree_{p^*}$ is that, upon
the insertion of a new point~$p$ (or the deletion of an existing point),
many of the solutions change in the same way. To formalize this,
let $\Delta_j$ be the signed difference of the cost of the range assignment $\Gamma(P,p^*,\lambda_j)$
before and after the insertion of $q$, where $\Delta_j$ is positive if the cost increases. 
Figure~\ref{fig:update} shows various possible values for $\Delta_j$, depending on the location
of the new point~$q$ with respect to the range~$\lambda_j$. 
\begin{figure}
\begin{center}
\includegraphics{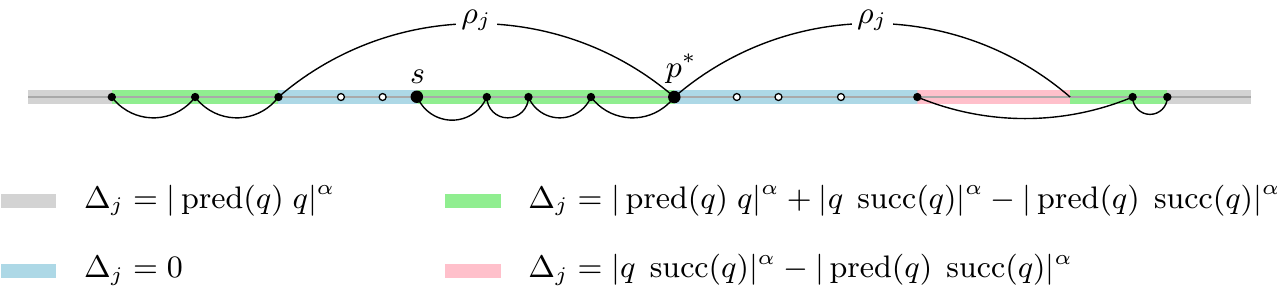}
\end{center}
\caption{Various cases that can arise when a new point~$q$ is inserted into~$P$.
Open disks indicate zero-range points. The arcs indicate the ranges of the points
before the insertion of~$q$, where the range of the root-crossing point is drawn both to its right and to its left.
The colored intervals relate the possible locations of~$q$ to the corresponding values~$\Delta_j$.}
\label{fig:update}
\end{figure}
As can be seen in the figure, 
there are only four possible values for $\Delta_j$.
This allows us to design our data structure $\tree_{p^*}$
such that it can be updated using $O(1)$ 
\emph{bulk updates} of the following form: 
\begin{quotation}
\noindent Given an interval $I$ of range values and an update value~$\Delta$, add $\Delta$ to the
cost of $\Gamma(P,p^*,\lambda_j)$ for all $\lambda_j\in I$.
\end{quotation}
In Appendix~\ref{app:data-structure} we define the information stored
in~$\tree_{p^*}$ and we
show how bulk updates can be done in $O(\log n)$ time.
We eventually obtain the following theorem.
\begin{restatable}{theorem}{broadcastupdate}
\label{thm:broadcast-update}
An optimal solution to the broadcast range-assignment problem 
for a point set $P$ in~$\Reals^1$ can be maintained in $O(n\log n)$ per insertion and deletion,
where $n$ is the number of points in the current set~$P$.
\end{restatable}

\section{A stable approximation scheme in $\Reals^1$}
\label{sec:sas}
In this section we use the structure of an optimal solution provided by
Theorem~\ref{thm:structure-1D} to obtain a SAS for the 1-dimensional broadcast 
range-assignment problem. Our SAS has stability parameter~$k(\eps)= O((1/\eps)^{1/(\alpha-1)})$,
which we will show to be asymptotically optimal.

\medskip
The optimal range assignment can be very unstable. Indeed, suppose the current point set
is $P := \{s,r_1,\ldots,r_n\}$ with $s=0$ and $r_i=i$ ($1\leq i \leq n$), and take any $\alpha > 1$. 
Then the (unique) optimal assignment~$\ropt$
has $\ropt(s)=\ropt(r_1)=\cdots=\ropt(r_{n-1})=1$ and $\ropt(r_n)=0$.
If now the point $\ell_1=-n$ is inserted, then the optimal assignment becomes
$\ropt(s)=n$ and $\ropt(r_1)=\cdots=\ropt(r_n)=\ropt(\ell_1)=0$, 
causing~$n$ ranges to be modified.

Next, we will define a feasible solution, referred to as a {\em canonical range assignment}~$\rk$ that is more stable than an optimal
assignment, while still having a cost close to the cost of an optimal solution.
Here $k$ is a parameter that allows a trade-off between stability and quality of the solution.
The assignment~$\rho_k$ for a given point set~$P$ will be uniquely determined by the 
set~$P$--it does not depend on the order in which the points have been inserted or deleted.
This means that the update algorithm simply works as follows. 
Let $\rho_k(P)$ be the canonical range assignment for a point set~$P$,
and suppose we update $P$ by inserting a point~$q$. 
Then the update algorithm computes~$\rho_k(P\cup\{q\})$ and it modifies
the range of each point~$p\in P\cup\{q\}$ whose canonical range
in $\rk(P\cup \{q\})$ is different from its canonical range in~$\rk(P)$.
The goal is now to specify $\rho_k$ such that (i) many ranges in $\rk(P\cup \{q\})$ are the same as in~$\rk(P)$, 
(ii) the cost of $\rk(P)$ is close to the cost of $\ropt(P)$.

The instance in the example above shows that there can be many points whose range changes from being standard to
being zero (or vice versa) when preserving optimality of the consecutive instances. Our idea is therefore to construct solutions where the number of points with zero range is limited, and instead give many points their standard range; if we do this for points whose standard range is relatively small, then the cost of this solution remains bounded compared to the cost of an optimum solution.
We now make this idea precise. 
\medskip

Consider a point set $P$ and let $\ropt$ be an optimal range assignment satisfying the structure described in Theorem~\ref{thm:structure-1D}. Assuming there are points in~$P$ on both 
sides of the source, $\ropt$ induces a broadcast tree~$\B$ with the 
structure depicted in Figure~\ref{fig:optimal-structure}. Let $\rst(p)$
be the standard range of a point~$p$.
The canonical range assignment~$\rho_k$ is now defined as follows.
\begin{itemize}
\item If all points from $P$ lie to the same side of $s$, then $\rho_k(p) := \ropt(p)$
      for all $p\in P$. Note that in this case $\rho_k(p)=\rst(p)$ for all $p\in P$.
\item Otherwise, let $Z
$ be the set of zero-range points in $\ropt(P)$. 
      If $|Z|\leq k$ then let $Z_k := Z$;
      otherwise let $Z_k\subseteq Z$ be the $k$~points from $Z$ with the largest standard ranges, 
      with ties broken arbitrarily. We define $\rk$ as follows. 
      \begin{itemize}
      \item $\rk(p) := \ropt(p)$ for all $p\in P\setminus Z$. Observe that this means that $\rk(p)=\rst(p)$ for 
                        all $p\in P\setminus Z$ except (possibly) for the root-crossing point.
      \item $\rk(p) := 0$ for all $p\in Z_k$. 
      \item $\rk(p) := \rst(p)$ for all $p\in Z\setminus Z_k$.
      \end{itemize}
\end{itemize}
Notice that $\rk$ is a feasible solution since $\rk(p) \geq \ropt(p)$ for each $p \in P$. 
The next lemma analyzes the stability of the canonical range assignment~$\rk$.
Recall that for any range assignment~$\range$---hence, also for~$\rk$---and any point $q$ 
not in the current set $P$, we have $\range(q)=0$ by definition.
The proof of the following lemma is in the appendix.
\begin{restatable}{lemma}{sasstability}
\label{le:sas-stability}
Consider a point set $P$ and a point~$q\not\in P$. Let $\rold(p)$ be the range of a point~$p$
in $\rk(P)$ and let $\rnew(p)$ be the range of~$p$ in $\rk(P\cup \{q\})$. Then 
\[
\left| \{ p\in P\cup\{q\} : \rnew(p) > \rold(p) \} \right| \leq k+3
\mbox{ and }
\left| \{ p\in P\cup\{q\} : \rnew(p) < \rold(p)  \} \right| \leq k+3.
\]
\end{restatable}
Next we bound the approximation ratio of~$\rho_k$.
\begin{lemma}\label{le:sas-approx}
For any set $P$ and any $\alpha > 1$, we have
$\cost_{\alpha}(\rho_k(P)) \leq \left(1+\frac{2^\alpha}{k^{\alpha-1}}\right)\cdot \cost_{\alpha}(\ropt(P))$.
\end{lemma}
\begin{proof}
If all points in $P$ lie to the same side of~$s$ then $\rk(P)=\ropt(P)$, and we are done.
Otherwise, let $p^*$ be the root-crossing point. 
The only points receiving a different range in~$\rk(P)$ when compared to $\ropt(P)$ are the 
points in~$Z\setminus Z_k$; these points have $\rk(p)=\rst(p)$ while $\ropt(p)=0$. 
This means we are done when $Z\setminus Z_k =\emptyset$. Thus we can assume that $|Z|>k$, 
so $Z\setminus Z_k \neq \emptyset$.
Assume without loss of generality that~$\ropt(p^*)=1$. As each $p \in Z$ is within 
reach of $p^*$, we have
$\sum_{p\in Z} \rst(p) \leq 2$.
Since $Z_k$ contains the $k$~points with the largest standard ranges among the points in $Z$,
we have $\max \{ \rst(p) : p\in Z\setminus Z_k \} \leq 2/k$. Hence,
\[
\sum_{p\in Z\setminus Z_k} \rk(p)^{\alpha}
 = \sum_{p\in Z\setminus Z_k} \rst(p)^{\alpha}
 = \sum_{p\in Z\setminus Z_k} \rst(p)^{\alpha-1} \cdot \rst(p)
 \leq \left(\frac{2}{k}\right)^{\alpha-1} \sum_{p\in Z\setminus Z_k}  \rst(p)
 \leq \frac{2^{\alpha}}{k^{\alpha-1}}.
\]
(The analysis can be made tighter by using that 
$\sum_{p\in Z\setminus Z_k}  \rst(p) \leq 2 - k \max_{p\in Z\setminus Z_k}\rst(p)$,
but this will not change the approximation ratio asymptotically.)
We conclude that
\[
\frac{\cost_{\alpha}(\rk(P))}{\cost_{\alpha}(\ropt(P))}   
\leq  \frac{\sum_{p\in P\setminus(Z\setminus Z_k)} \rk(p)^{\alpha}+\sum_{p\in Z\setminus Z_k} \rk(p)^{\alpha}}{\sum_{p\in P\setminus(Z\setminus Z_k)} \ropt(p)^{\alpha}} 
\leq 1+ \frac{2^{\alpha}}{k^{\alpha-1}},
\]
where the last inequality follows because we have $\rk(p)=\ropt(p)$ for all 
$p\in P\setminus(Z\setminus Z_k)$ and $\sum_{p\in P\setminus(Z\setminus Z_k)} \ropt(p)^{\alpha} \geq 1$.
\end{proof}
By maintaining the canonical range assignment~$\rk$ for 
$k=(2^{\alpha}/\eps)^{1/(\alpha-1)}=O((1/\eps)^{1/(\alpha-1)})$ we obtain the following theorem. 
\begin{restatable}{theorem}{sasalg}
\label{thm:sas-alg}
There is a SAS for the dynamic broadcast range-assignment problem 
in $\Reals^1$ with stability parameter~$k(\eps) = O((1/\eps)^{1/(\alpha-1)})$,
where $\alpha > 1$ is 
the distance-power gradient.
The time needed by the SAS to compute the new range assignment upon the insertion or
deletion of a point is~$O(n\log n)$, where $n$ is the number of points in the current set.
\end{restatable}

Next we show that the stability parameter $k(\eps)$ in our SAS is asymptotically optimal.
\begin{theorem} \label{thm:sas-lb}
Any SAS for the dynamic broadcast range-assignment problem in $\Reals^1$ 
must have stability parameter~$k(\eps) = \Omega((1/\eps)^{1/(\alpha-1)})$, where $\alpha > 1$ is 
the distance-power gradient.
\end{theorem}
\begin{proof}
Let \alg be a $k$-stable algorithm, where $k\geq 4$ and $k^{\alpha-1} \geq \frac{1}{2^{\alpha+1}(2^{\alpha-1}-1)}$
and $k$ is even, and let $\ralg$ be the range assignment 
it maintains. Note that the condition on $k$ is satisfied for $k$ large enough.
We will show that the approximation ratio of \alg is 
at least~$1+\frac{1}{2^{\alpha+2} k^{\alpha-1}}$.
Since a SAS has approximation ratio $1+\eps$, this implies
that the stability parameter $k(\eps)$ of \alg must satisfy $k(\eps)=\Omega((1/\eps)^{1/(\alpha-1)})$.

Consider the point set $P := \{s,r_1,r_2,\ldots r_{2k}\}$, where $s=0$ and $r_i = i/(2k)$ 
for $i=1,2, \ldots, 2k$. We consider two cases.
\\[2mm]
\emph{Case I: The number of zero-range points in $\ralg(P)$ is at least~$k/2$,}
     where we assume without loss of generality that all points with range
      less than~$1/(2k)$ actually have range zero. 
      The cheapest possible solution in this case is to have exactly $k/2$ zero-range points,
      $k$ points with range~$1/(2k)$, and $k/2$ points with range $1/k$, for a total cost
      of
      \[
      \cost_{\alpha}(\ralg(P)) \geq k \cdot \left(\frac{1}{2k}\right)^{\alpha} + \frac{k}{2} \cdot \left(\frac{1}{k}\right)^{\alpha}
      = \left(1+\frac{2^{\alpha-1}-1}{2}\right)\cdot 2k \left(\frac{1}{2k}\right)^{\alpha}.
      \]
      An optimal solution has cost~$2k\cdot (1/(2k))^{\alpha}$, 
      and so the approximation ratio of \alg in Case~I is at least~$1+\frac{2^{\alpha-1}-1}{2}$, 
      which is at least $1+\frac{1}{2^{\alpha+2} k^{\alpha-1}}$ since  
      $k^{\alpha-1} \geq \frac{1}{2^{\alpha+1}(2^{\alpha-1}-1)}$.
      \\[2mm]
\emph{Case II: The number of  zero-range points $\ralg(P)$ is less than~$k/2$.}
      Now suppose the point $\ell_1 = -1$ arrives. Since $\ralg(P)$
      had less than $k/2$ zero-range points and \alg can modify at most~$k$ ranges, 
      $\ralg(P\cup\{\ell_1\})$ has less than $3k/2$ zero-range points. Hence, at least 
      $k/2$ points in $P\cup\{\ell_1\}$ have a range that is at least~$1/(2k)$,
      one of which must have a range at least~1. This 
      implies that
      $
      \cost_{\alpha}(\ralg(P\cup\{\ell_1\})) \geq 1 + (k/2-1) \cdot \left(\frac{1}{2k}\right)^{\alpha} 
      \geq 1+\frac{1}{2^{\alpha+2} k^{\alpha-1}},
      $
      where the last inequality holds since $k/2-1\geq k/4$ (because $k\geq 4$).
      An optimal range assignment on $P\cup\{\ell_1\}$
      has $\ropt(s)=1$ and all other ranges equal to zero, for a total cost of~1,
      and so the approximation ratio of \alg in Case~II is at least
      $1+\frac{1}{2^{\alpha+2} k^{\alpha-1}}$ as well.
\end{proof}

\section{The problem in $\bbS^1$}
\label{sec:S1}
We now turn to the setting where the underlying space is $\bbS^1$, that is,
the points in $P$ lie on a circle and distances are measured along the circle. In Section~\ref{subsec:S1-structure}, we prove that the structure of an optimal solution in $\bbS^1$
is very similar to the structure of an optimal solution in~$\Reals^1$ as formulated in Theorem~\ref{thm:structure-1D}. In spite of this, and contrary to the problem in $\Reals^1$, we prove in Section~\ref{subsec:S1-no-SAS} that no SAS exists for the problem in $\bbS^1$.

Again, we denote the source point by~$s$.
The clockwise distance from a point $p\in \bbS^1$ to a point~$q\in \bbS^1$
is denoted by $\dcw(p,q)$, and
the counterclockwise distance by $\dccw(p,q)$. The actual distance
is then $d(p,q) := \min(\dcw(p,q),\dccw(p,q))$. 
The closed and open clockwise interval from $p$ to $q$ are denoted by~$[p,q]^{\cw}$
and $(p,q)^{\cw}$, respectively.

\subsection{The structure of an optimal solution in $\bbS^1$}
\label{subsec:S1-structure}
Here we prove that the structure of an optimal solution in $\bbS^1$
is very similar to the structure of an optimal solution in~$\Reals^1$.
The heart of this proof is the following lemma, whose (rather intricate)
proof is given in Appendix~\ref{app:structure-in-S1}.
Define  the \emph{covered region} of $P$ with respect to 
a range assignment $\rho$, denoted by $\cov(\rho,P)$, to be the set of all points $r\in \bbS^1$ 
such that there 
exists a point $p\in P$ with $\rho(p)\geq d(p,r)$. 
\begin{restatable}{lemma}{Sstructure}
\label{lem:S1-structure}
Let $P$ be a point set in $\bbS^1$ with $|P|>2$ and let $\ropt$ be an optimal range assignment for $P$.
Then there exists a point $r\in \bbS^1$ such that $r\notin \cov(\ropt,P)$.
\end{restatable}
Lemma~\ref{lem:S1-structure} implies that an optimal solution for an instance 
in $\bbS^1$ corresponds to an optimal solution for an instance in~$\Reals^1$
derived as follows. 
For a point $r\in \bbS^1$, define the mapping $\mu_r:P\rightarrow \Reals^1$ such that $\mu_r(s) := 0$,
and $\mu_r(p) := \dcw(s,p)$ for all $p\in [s,r]^{\cw}$,
and $\mu_r(p) := - \dccw(s,p)$ for all $p\in [r,s]^{\cw}$.
Let $\mu_r(P)$ denote the resulting point set in $\Reals^1$.
\begin{theorem} \label{thm:opt-in-S1}
Let $P$ be an instance of the broadcast range-assignment problem in~$\bbS^1$.
There exists a point~$r\in \bbS^1$ such that an optimal range assignment for $\mu_r(P)$ 
in $\Reals^1$ induces an optimal range assignment for $P$.
Moreover, we can compute an optimal range assignment for~$P$
in $O(n^2 \log n)$ time, where $n$ is the number of points in $P$.
\end{theorem}
\begin{proof}
Let $r\in \bbS^1$ be a point such that $r\notin \cov(\ropt,P)$, which exists by
Lemma~\ref{lem:S1-structure}. Consider
the mapping $\mu_r$. Any feasible range assignment for $\mu_r(P)$ induces a 
feasible range assignment for $P$ in $\bbS^1$, since $d(p,q)\leq |\mu_r(p)\mu_r(q)|$ 
for any two points~$p,q\in P$.
Conversely, an optimal range assignment for $P$ induces a feasible
range assignment for $\mu_r(P)$, since the point $r$ is not covered in the optimal solution.
This proves the first part of the theorem.

Now let $P:=\{s,p_1,\ldots,p_n\}$, where the points $p_i$ are ordered 
clockwise from~$s$. For $0\leq i\leq n$, let $r_i$ be a point in~$(p_i,p_{i+1})^{\cw}$,
where $p_0 = p_{n+1} = s$. Since $\mu_{r_i}=\mu_r$ for any $r\in(p_i,p_{i+1})^{\cw}$, 
an optimal solution can be computed by finding the best solution over all mappings~$\mu_{r_i}$.
The only difference between $\mu_{r_i}$ and $\mu_{r_{i+1}}$ 
is the location that~$p_{i+1}$ is mapped to, so after computing an optimal solution for $\mu_1(P)$ in $O(n^2\log n)$ time, 
we can go through the mappings $\mu_2,\ldots,\mu_{n}$
and update the optimal solution in $O(n\log n)$ time using 
Theorem~\ref{thm:broadcast-update}. Hence, an optimal range assignment
for $P$ can be computed in $O(n^2 \log n)$ time.
\end{proof}

\subsection{Non-existence of a SAS in $\bbS^1$}
\label{subsec:S1-no-SAS}
We have seen that an optimal solution for a set $P$ in~$\bbS^1$ can be obtained from 
an optimal solution in~$\Reals^1$, if we cut~$\bbS^1$ at an appropriate point~$r$.
It is a fact however that the insertion of a new point into~$P$ may cause the location of the cutting 
point~$r$ to change drastically. Next we show that this means that
the dynamic problem in $\bbS^1$ does not admit a SAS. 
\begin{theorem} \label{thm:S1-no-SAS}
The dynamic broadcast range-assignment problem in~$\bbS^1$ with distance power gradient~$\alpha>1$ does not admit a SAS. In particular,
there is a constant $\ca>1$ such that the following holds: for any $n$ large enough, there is 
a set $P :=\{s,p_1,\ldots,p_{2n+1}\}$ and a point $q$ in $\bbS^1$ such that any update algorithm~\alg
that maintains a $\ca$-approximation must modify more than $2n/3-1$ ranges upon the insertion of~$q$ into~$P$.
\end{theorem}
The rest of this section is dedicated to proving Theorem~\ref{thm:S1-no-SAS}.
We will prove the theorem for 
\[
\ca := \min \left(  1 + 2^{\alpha-4} - \frac{1}{8},  \
                    1 + \frac{2^{\alpha-1}-1}{3\cdot 2^{\alpha}+2}, \   
                     1 + \frac{\min \left( 2^{\alpha}-1, \frac{3^{\alpha}-2^{\alpha}-1}{2}, \frac{ 4^{\alpha}-2^{\alpha}-2}{3} \right)}{4(2^{\alpha}+1)}
            \right).
\]
Note that each term is a constant strictly greater than~1 for any fixed constant~$\alpha>1$.
In particular, for $\alpha=2$ we have $\ca = 1+\frac{1}{14}$.
\medskip

Let $P:=\{s,p_1,\ldots,p_{2n+1}\}$, where $\dcw(p_i,p_{i+1})=2$ for odd~$i$ 
and $\dcw(p_i,p_{i+1})=1$ for even~$i$; see Fig.~\ref{fig:no-sas-cirle-new}(i). 
Let $\dcw(s,p_1) = \delta$, where $\delta^{\alpha} = (2^{\alpha}+1)n$. 
Finally, let $\dcw(p_{2n+1},q) = \dcw(q,s) = x\delta$,
where  $x^{\alpha}= \frac{1}{4}+\left(\frac{1}{2}\right)^{\alpha+1}$.
Note that $(1/2)^{\alpha} < x^{\alpha} < 1/2$
for any~$\alpha>1$.
\begin{figure}
\begin{center}
\includegraphics{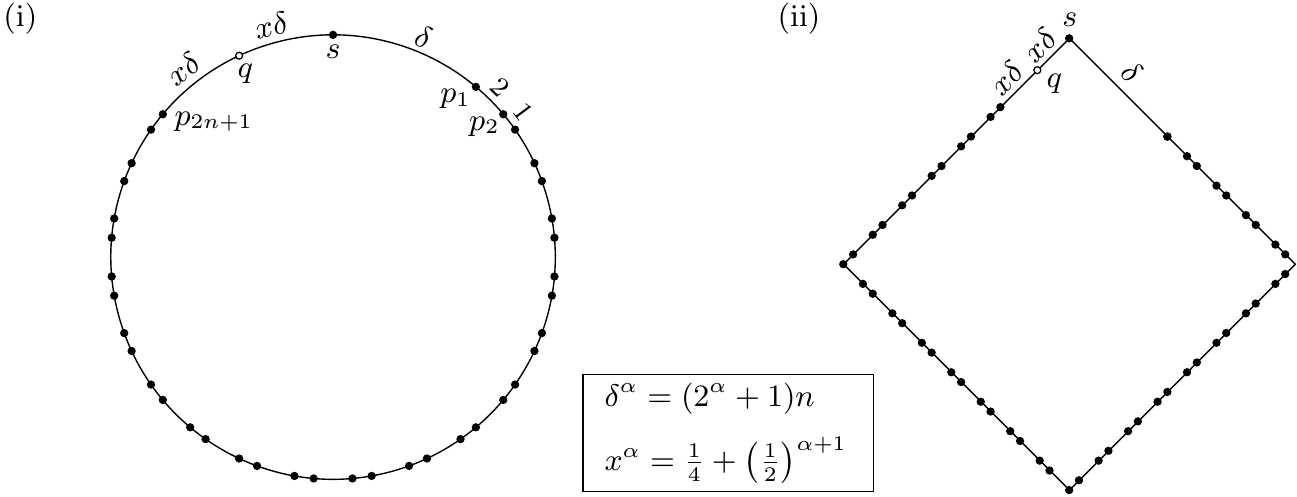}
\end{center}
\caption{(i) The instance showing that there is no SAS in~$\bbS^1$.
(ii) The instance in~$\Reals^2$.}
\label{fig:no-sas-cirle-new}
\end{figure}

Let $\range(p)$ denote the range given to a point $p$ by~\alg.
A directed edge $(p,p')$ in the communication graph induced by~$\range$
is called a \emph{clockwise edge} if $\range(p)\geq \dcw(p,p')$, and it is called
a \emph{counterclockwise edge} 
if $\range(p)\geq \dccw(p,p')$. Observe that we may assume that no edge $(p,p')$ is both 
clockwise and counterclockwise, because otherwise $\range(p) \geq (\delta+3n+2x\delta)/2$, 
which is much too expensive for an approximation ratio of at most~$\ca$.
Define the range $\range(p)$ of a point in~$P$ to be \emph{\CW-minimal} 
if $\range(p)$ equals the distance from $p$ to its clockwise neighbor in~$P$. 
Similarly, $\range(p)$ is \emph{\CCW-minimal} if $\range(p)$ equals the distance 
from $p$ to its counterclockwise neighbor. The idea of the proof is to show that
before the insertion of~$q$, most of the points~$s,p_1,\ldots,p_{2n+1}$
must have a \CW-minimal range, while after the insertion most points must
have a \CCW-minimal range. This will imply that many ranges must 
be modified from being \CW-minimal to being \CCW-minimal. 
\medskip

Before the insertion of~$q$, giving every point a \CW-minimal range leads to a 
feasible assignment of total cost $\delta^{\alpha} + (2^{\alpha}+1)n = 2\delta^\alpha$. 
After the insertion of~$q$, giving every point a \CCW-minimal range leads to a 
feasible assignment of total cost $2(x\delta)^{\alpha} + (2^{\alpha}+1)n = (2x^{\alpha}+1)\delta^\alpha$. 
Hence, if $\opt(\cdot)$ denotes the cost of an optimal range assignment,
then we have:
\begin{observation} \label{obs:S1-opt}
$\opt(P)\leq 2\delta^{\alpha}$ \ \ and \ \
$\opt(P\cup\{q\})\leq (2x^{\alpha}+1)\delta^{\alpha} < 2\delta^{\alpha}$.
\end{observation}
We first prove a lower bound on the total cost of the points~$p_1,\ldots,p_{2n+1}$.
Intuitively, only $o(n)$ of those points can be reached from~$s$ or $q$ (otherwise the range
of $s$ or $q$ would be too expensive) and the cheapest way to reach
the remaining points will be to use only \CW-minimal or \CCW-minimal
ranges. A formal proof of the lemma is given in Appendix~\ref{app:no-sas-missing-proofs}.
\begin{restatable}{lemma}{costofpi}
\label{lem:cost-of-pi}
$\sum_{i=1}^{2n+1}\rho(p_i)^\alpha \geq (2^{\alpha}+1)n - o(n)$, both before and after the insertion of~$q$.
\end{restatable}
\noindent The following lemma gives a key property of the construction.
\begin{restatable}{lemma}{nocworccwedge}
\label{lem:no-cw-or-ccw-edge}
The point~$p_{2n+1}$ cannot have an incoming counterclockwise edge before $q$ is inserted,
and the point~$p_{1}$ cannot have an incoming clockwise edge after $q$ has been inserted.
\end{restatable}
\begin{proof}
The cheapest incoming counterclockwise edge for~$p_{2n+1}$ before the insertion of~$q$
is from~$s$, but this is too expensive for \alg to achieve approximation ration~$\ca$.
Similarly, the cheapest incoming clockwise edge for~$p_1$ is from $s$,
but this is too expensive after the insertion of~$q$.
The computations for both cases can be found Appendix~\ref{app:no-sas-missing-proofs}. 
\end{proof}
We are now ready to prove that many edges must change from being \CW-minimal
to being \CCW-minimal when~$q$ is inserted.

\begin{lemma} \label{lem:many-minimal-ranges}
Before the insertion of~$q$, at least $4n/3+1$ of the points from $\{s,p_1,\ldots,p_{2n}\}$ 
have  a \CW-minimal range and after the insertion of~$q$ at least $4n/3+1$ of the points 
from $\{q,p_1,\ldots,p_{2n}\}$ have  a \CCW-minimal range.
\end{lemma}
\begin{proof}
We prove the lemma for the situation before $q$ is inserted; the proof for the situation
after the insertion of~$q$ is similar. Observe that before and after the insertion of~$q$, the distance
between any two points is either~1,~2 or at least~3. Hence, in what
follows we may assume
that $\range(p)\in\{0,1,2\}\cup [3,\infty)$ for any point~$p \in P \cup\{q\}$.

It will be convenient to define $p_0 := s$ (although we may still use~$s$
if we want to stress that we are talking about the source). 
Recall that $p_{2n+1}$ does not have an incoming counterclockwise edge in the communication 
graph~$\cg(P)$ before the insertion of~$q$.
Let $\pi^*$ be a minimum-hop path from $s$ to $p_{2n+1}$ in~$\cg(P)$. 
Since $p_{2n+1}$ does not have an incoming counterclockwise edge and $\pi^*$
is a minimum-hop path, all edges in $\pi$ are clockwise.
We assign each point $p_j$ with $1\leq j\leq 2n+1$ to the edge
$(p_i,p_t)$ in $\pi^*$ such that $i+1\leq j\leq t$, and we define
$A(p_i,p_t):=\{p_{i+1},\ldots,p_t\}$ to be the set of all points assigned to~$(p_i,p_t)$.
We define the \emph{excess} of a point~$p_j\in A(p_i,p_t)$ to be
\[
\exc(p_j) := \frac{1}{|A(p_i,p_t)|} \cdot \left( \range(p_i)^{\alpha} - \sum_{p_{\ell}\in A(p_i,p_t)} d(p_{\ell-1}, p_{\ell})^{\alpha} \right).
\]
We say that an edge $(p_i,p_t)$ in $\pi^*$ is \CW-minimal
if $p_i$ has a \CW-minimal range.
Note that if a point~$p_j$ is assigned to
a \CW-minimal edge, then this is the edge~$(p_{j-1},p_j)$ and 
$\exc(p_j)=0$. Intuitively, $\exc(p_j)$ denotes the additional cost we
pay for reaching~$p_j$ compared to reaching it by a \CW-minimal edge,
if we distribute the additional cost of a non-\CW-minimal edge over
the points assigned to it. Because each of the points $p_1,\ldots,p_{2n+1}$
is assigned to exactly one edge on the path $\pi^*$,
we have
\begin{equation}  \label{eq:total-cost-with-excess}
\sum_{p_i\in \pi^*} \range(p)^{\alpha}
\geq \sum_{j=1}^{2n+1} d(p_{j-1}, p_{j})^{\alpha}  +  \sum_{j=1}^{2n+1} \exc(p_j)
\geq \opt(P) +  \sum_{j=1}^{2n+1} \exc(p_j)
\end{equation}
where the second inequality follows from Observation~\ref{obs:S1-opt} and
because $p_0=s$. The following claim is proved in Appendix~\ref{app:no-sas-missing-proofs}.
(Essentially, the smallest possible excess is obtained when
$|A(p_i,p_t)|\in \{1,2,3\}$; the three terms in the claim correspond to these cases.)
\begin{claiminproof}
If $p_j$ is not assigned to a \CW-minimal edge then $\exc(p_j)\geq c'_{\alpha}$,
where $c'_{\alpha} = \min \left( 2^{\alpha}-1, \frac{3^{\alpha}-2^{\alpha}-1}{2}, \frac{ 4^{\alpha}-2^{\alpha}-2}{3}  \right)$.
\end{claiminproof}
Now suppose for a contradiction that less than $4n/3+1$ points from 
$\{s,p_1,\ldots,p_{2n+1}\}$ have a \CW-minimal range. Then at least $2n/3+1$ points~$p_j$
have $\exc(p_j) \geq c'_{\alpha}$ by the claim above. By Inequality~(\ref{eq:total-cost-with-excess})
the total cost incurred by \alg is therefore more than
\begin{eqnarray}
\opt(P) + c'_{\alpha} \cdot (2n/3) 
& = & \opt(P) +  \frac{c'_{\alpha}}{3(2^{\alpha}+1)} \cdot 2(2^{\alpha}+1) n   \label{eq2} \\
& > & \left( 1 + \frac{\min \left( 2^{\alpha}-1, \frac{3^{\alpha}-2^{\alpha}-1}{2}, \frac{ 4^{\alpha}-2^{\alpha}-2}{3}  \right)}{4(2^{\alpha}+1)} \right) \cdot \opt(P) \label{eq3} \\[1mm]
& \geq & \ca \cdot \opt(P) \label{eq4}
\end{eqnarray}
which contradicts the approximation ratio achieved by~\alg.
\end{proof}
Lemma~\ref{lem:many-minimal-ranges} implies that at least $4n/3$ of the points $p_1,\ldots,p_{2n+1}$
have a \CW-minimal range before $q$ is inserted, and at least $4n/3$ of 
those points have a \CCW-minimal range after the insertion.
Hence, at least $2n+1 - 2 \cdot (2n/3+1) = 2n/3 -1$ points must change from being
\CW-minimal to being \CCW-minimal, thus finishing the proof of Theorem~\ref{thm:S1-no-SAS}.
\section{The 2-dimensional problem}
\label{sec:higher-dim}
The broadcast range-assignment problem is \np-hard in $\Reals^2$, so we cannot expect a 
characterization of the structure of an optimal solution similar to Theorem~\ref{thm:structure-1D}.
Using a similar construction as in~$\bbS^1$
we can also show that the problem in $\Reals^2$ does not admit a~SAS.
\begin{restatable}{theorem}{noSASplane}
\label{thm:no-SAS-R2}
The dynamic broadcast range-assignment problem in~$\Reals^2$ with distance power gradient~$\alpha>1$ 
does not admit a SAS. In particular, there is a constant $\ca>1$ such that the following holds: 
for any $n$ large enough, there is  a set $P :=\{s,p_1,\ldots,p_{2n+1}\}$ and a point $q$ in $\Reals^2$ 
such that any update algorithm~\alg
that maintains a $\ca$-approximation must modify at least $2n/3-1$ ranges upon the insertion of~$q$ into~$P$.
\end{restatable}
\begin{proof}
We use the same construction as in $\bbS^1$, where we embed the points on a square and
the distances used to define the instance are measured along the square; see Fig.~\ref{fig:no-sas-cirle-new}(ii). 
For any two points on the same edge of the square, their distance along the square 
is the same as their distance in~$\Reals^2$. Moreover, we know that no range can be larger
than~$3\delta$, otherwise the range assignment is already too expensive. Note that the
number of points $p_i$ within distance $3\delta$ from a corner is $o(n)$.
Using this fact we can argue that all lemmas from Section~\ref{subsec:S1-no-SAS} still hold.
(For example, in the proof of Lemma~\ref{lem:many-minimal-ranges} we can simply ignore the
excess of $o(n)$ points.) This is argued more extensively in Appendix~\ref{app:no-SAS-R2}.
\end{proof}
Although the problem in~$\Reals^2$ does not admit a SAS, there is a relatively simple 
$O(1)$-stable $O(1)$-approximation algorithm for~$\alpha \geq 2$.
The algorithm is based on a result by Amb\"uhl~\cite{DBLP:conf/icalp/Ambuhl05}, 
who showed that a minimum spanning tree (MST) on~$P$ gives a 6-approximation 
for the static broadcast range-assignment problem:
turn the MST into a directed tree rooted at the source~$s$, and assign 
as a range to each point $p\in P$  the maximum length of any of its outgoing edges.
To make this solution stable, we set the range of each
point to the maximum length of any of its incident edges (not just the outgoing ones).
Because an MST in $\Reals^2$ has maximum degree~6, this leads to 
17-stable 12-approximation algorithm; see Appendix~\ref{app:mst-in-R2}.

\section{Concluding remarks}
\label{sec:conclusions}
%
We studied the dynamic broadcast range-assignment problem from a stability perspective,
introducing the notions of $k$-stable algorithms and stable approximation schemes (SASs).
Our results provide a fairly complete picture of the problem in $\Reals^1$, in $\bbS^1$,
and in $\Reals^2$. In particular, we presented a SAS in~$\Reals^1$ that has an
asymptotically optimal stability parameter, and showed that the problem does
not admit a SAS in $\bbS^1$ and $\Reals^2$. Future work can focus on improving 
(the upper and/or lower bounds for) the approximation ratios we have obtained for 
algorithms with constant stability parameter. In particular, 
it is open whether there exists a 2-stable algorithm with approximation ratio less than~2 in $\Reals^1$.
It would also be interesting to use develop algorithm with small stability parameter
in~$\bbS^1$, possibly using the relation we proved between the structure of an optimal
structure in in $\bbS^1$ and in $\Reals^1$.



\bibliography{references}

\newpage

\appendix
\section{A data structure for maintaining an optimal solution}
\label{app:data-structure}
Below we give a detailed description of the data structure to maintain an optimal solution 
in~$\Reals^1$. For convenience, we repeat somematerial from the main text, so that the
description below is self-contained.
\medskip

Recall that an optimal solution for a given point set~$P$ has a single 
root-crossing point,~$p^*$. Once the range of $p^*$ is fixed, the solution 
is completely determined. The range of~$p^*$
is defined by some other point from $P$---we have $\range(p^*)=|p^* p|$ for
some point $p\neq p^*$---and so there are $n-1$ candidate ranges for a given choice of the 
root-crossing point~$p^*$. The idea of our solution is to implicitly
store the cost of the range assignment for each candidate range of $p^*$ such that, 
upon the insertion or deletion of a point in~$P$,
we can in $O(\log n)$ time find the best range for~$p^*$. By maintaining
$n$~such data structures~$\tree_{p^*}$, one for each choice of the
root-crossing point~$p^*$, we can then find the overall best solution.

Besides the data structures $\tree_{p^*}$ which are described below,
we also maintain a global data structure~$\tree_P$ that supports the following operations.
\begin{itemize}
\item Find the predecessor $\prd(q)$ and successor $\suc(q)$ in $P$ of a query point~$q$.
\item Given two points $p,p'\in P$, report the total cost of the chain from~$p$ to~$p'$.
\item Insert or delete points from~$P$.
\end{itemize}
By implementing~$\tree_P$ as a suitably augmented binary search tree---see the book 
by Cormen~\etal~\cite[Chapter~15]{intro-to-alg} for the design and maintenance of such structures---each
of these operations can be performed in $O(\log n)$ time.  

\mypara{The data structure for a given root-crossing point.}
Next we explain our data structure for a given candidate root-crossing point~$p^*$.
We assume without loss of generality that $p^*$ lies to the right of
the source point~$s$; it is straightforward to adapt the structure to the 
(symmetric) case where $p^*$ lies to the left of~$s$, and 
to the case where $p^*=s$.

Let $\R_{p^*}$ be the set of all ranges we need to consider for~$p^*$, for the current set~$P$.
The range of a root-crossing point must extend beyond the source point. Hence,
\[
\R_{p^*} := \{ |p^* p| : p \in P \mbox{ and } |p^* p| > |p^* s| \}.
\]
Let $\lambda_1,\ldots, \lambda_m$ denote the sequence of ranges in~$\R_{p^*}$, ordered from small
to large. (If $\R_{p^*}=\emptyset$, there is nothing to do and our
data structure is empty.)  As mentioned, once we fix a range $\lambda_j$ for the given
root-crossing point~$p^*$, the solution is fully determined by Theorem~\ref{thm:structure-1D}:
there is a chain from $s$ to $p^*$, a chain from the rightmost point within range
of $p^*$ to the right-extreme point, and a chain from the leftmost point within range
of $p^*$ to the left-extreme point.
We denote the resulting range assignment\footnote{When all points in $P$ lies to the same side of~$s$,
then the range assignment is formally not root-crossing, but we will permit ourselves this slight 
abuse of terminology. Notice that in this case the range assignment induced by considering~$s$
as root-crossing point and setting $\range(s) := |s\suc(s)|$ gives a chain from $s$ to the extreme point
as solution, which is optimal.}
for $P$ by~$\Gamma(P,p^*,\lambda_j)$.

There is one subtlety in the definition of 
$\Gamma(P,p^*,\lambda_j)$, namely when there are no points within reach of $p^*$
to, say, the right of~$p^*$; see Figure~\ref{fig:dummy}.
Such a solution can never be optimal, but we must
maintain it nevertheless, because the range $\lambda_j$ may become relevant later.
To deal with this situation, we will insert a dummy point whose location coincides
with~$p^*$ and that is defined to be the predecessor of $\suc(p^*)$.
Note that the dummy will become a zero-range point as soon as an actual point 
is inserted that is within the range of $p^*$ and lies to the same side of $p^*$ as the dummy.
\begin{figure}
\begin{center}
\includegraphics{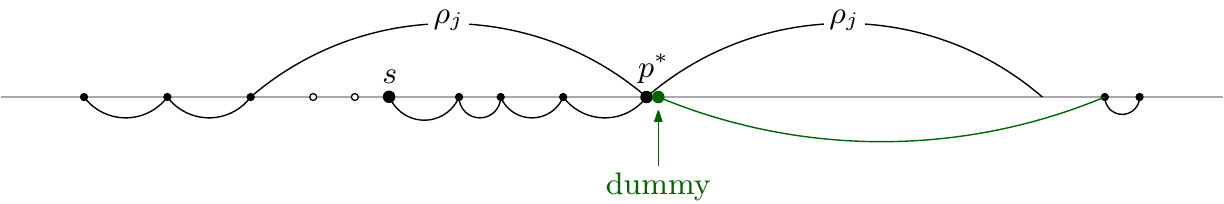}
\end{center}
\caption{A dummy point is inserted to ensure that points to the right of~$p^*$ can be reached.
For clarity the dummy is drawn to the right of~$p^*$, but it actually coincides with $p^*$.}
\label{fig:dummy}
\end{figure}
\medskip

Our data structure, which implicitly stores the costs of the range assignments 
$\Gamma(P,p^*,\lambda_j)$ for all $\lambda_j\in\R_{p^*}$, is an augmented balanced binary search tree~$\tree_{p^*}$, 
defined as follows.
\begin{itemize}
\item The leaves of $\tree_{p^*}$ are in one-to-one correspondence with the candidate ranges in $\R_{p^*}$:
    the leftmost leaf corresponds to $\lambda_1$, the next left to $\lambda_2$, and so on.
    From now on, with a slight abuse of notation, we use $\lambda_j$ to refer to a
    range in $\R_{p^*}$ as well as to the corresponding leaf.
\item Each leaf stores, besides the corresponding range~$\lambda_j$, 
    a value~$f(\lambda_j)$. Initially $f(\lambda_j)$ will be equal to the cost of 
    $\Gamma(P,p^*,\lambda_j)$. Later this may no longer be the case, however. 
\item The internal nodes of $\tree_{p^*}$ are augmented with extra information, as follows.
    For an internal node~$v$, let $\R_{p^*}(v) \subseteq \R_{p^*}$ be the set of all ranges
    stored in the leaves of the subtree rooted at~$v$. The node $v$ stores the
    following additional information, besides the splitting values that we have
    because $\tree_{p^*}$ is a search tree on the ranges in~$\R_{p^*}$:
    \begin{itemize}
    \item A \emph{correction value} $\Delta(v)\in \Reals$.
    \item A value $\mc(v)$ defined as follows. For  a range $\lambda_j\in \R_{p^*}(v)$
          define the \emph{local cost of $\lambda_j$ at $v$} to be $f(\lambda_j)+\sum_{u} \Delta(u)$,
          where the sum is over all nodes $u$ on the path from $v$ 
          (and including $v$) to $\lambda_j$.
          Then $\mc(v)$ is defined to be the minimum local cost over all ranges in $\R_{p^*}(v)$.
    \item A range $\lambda_j\in \R_{p^*}(v)$ whose local cost at $v$ is $\mc(v)$. This range is denoted by $\br(v)$.
    \end{itemize}
\end{itemize}

Our update algorithm  will ensure the following invariant:
\begin{equation}
\parbox{12cm}{For any range $\lambda_j\in \R_{p^*}$, the total cost of
    $\Gamma(P,p^*,\lambda_j)$ is equal to $f(\lambda_j)+\sum_{u} \Delta(u)$, where the sum is
             over all nodes on the search path from $\myroot(\tree_{p^*})$ to~$\lambda_j$.} \label{inv}
\end{equation}
In other words, Invariant~(\ref{inv}) states that, for any range $\lambda_j$, the local cost of $\lambda_j$
at the root of $\tree_{p^*}$ is equal to the actual cost of $\Gamma(P,p^*,\lambda_j)$.
Since $\R_{p^*}(\myroot(\tree_{p^*}))=\R_{p^*}$, this implies that $\mc(\myroot(\tree_{p^*}))$ equals the minimum cost
that can be obtained by a solution that uses~$p^*$
as a root-crossing point.

\mypara{Updating the data structure.}
We now describe how to update the structure upon the insertion of a new point.
Deletions can be handled in a symmetrical manner. 
To simplify the presentation, we assume that no two points in~$P$ coincide;
the solution is easily adapted to the case where $P$ can be a multi-set.
\begin{figure}
\begin{center}
\includegraphics{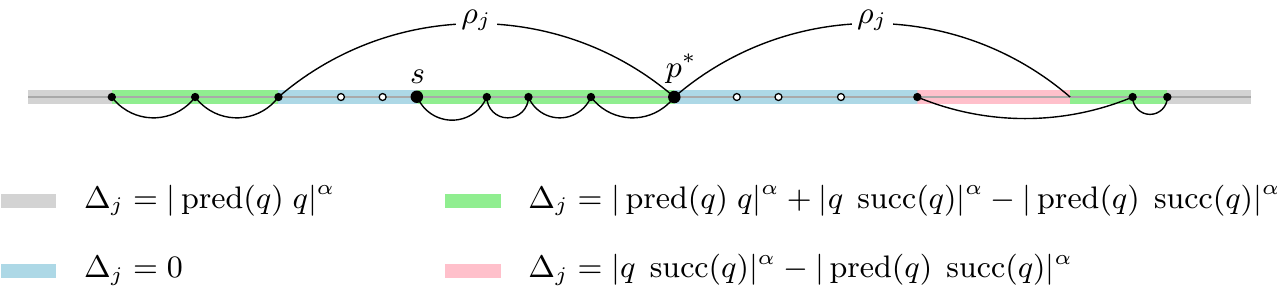}
\end{center}
\caption{Various cases that can arise when a new point~$q$ is inserted into~$P$.
Open disks indicate zero-range points. The arcs indicate the ranges of the points
before the insertion of~$q$, where the range of the root-crossing point is drawn both to its right and to its left.
The colored intervals relate the possible locations of~$q$ to the corresponding values~$\Delta_j$.}
\label{fig:update2}
\end{figure}
Let $\Delta_j$ be the (signed) difference of the cost of the range assignment $\Gamma(P,p^*,\lambda_j)$
before and after the insertion of $q$, where $\Delta_j$ is positive if the cost increases. 
Figure~\ref{fig:update2} shows various possible values for $\Delta_j$, depending on the location
of the new point~$q$. The figure is for the generic case, when $\Zleft,\Zright\neq\emptyset$
and there are points to the right as well as to the left of the interval that is within
reach of the root-crossing point~$p^*$. Lemma~\ref{le:cost-update},
which is easy to verify, gives the values for $\Delta_j$
for all cases, where we write $p<q$ when a point~$p$ is to the left of a point~$q$.
\begin{lemma}\label{le:cost-update}
Let $\Delta_j := \cost\left(\Gamma(P\cup \{q\},p^*,\lambda_j)\right) - \left(\Gamma(P,p^*,\lambda_j)\right)$. 
If $s<q<p^*$ or $p^*<q<s$ we have
\[
\Delta_j =  |\prd(q)\;q|^{\alpha} + |q\;\suc(q)|^{\alpha} - |\prd(q)\;\suc(q)|^{\alpha}
\]
Otherwise we have 
\[
\Delta_j = 
\left\{ \begin{array}{ll}
|q\;\suc(q)|^{\alpha} - |\prd(q)\;\suc(q)|^{\alpha} & \mbox{if  $|p^*q| \leq \rho_j < |p^*\;\suc(q)|$} \\[2mm]
0 & \mbox{if $\rho_j \geq  |p^*\;\suc(q)|$} \\
  & \hspace*{3mm} \mbox{or $ \big(\lambda_j \geq |p^* q|$ and $\suc(q)=\nil \big)$} \\[2mm]
|\prd(q)\;q|^{\alpha} + |q\;\suc(q)|^{\alpha} - |\prd(q)\;\suc(q)|^{\alpha} & \mbox{if $\lambda_j <|p^* q|$ and $\suc(q)\neq\nil$} \\[2mm]
|\prd(q)\;q|^{\alpha} & \mbox{if $\lambda_j <|p^* q|$ and $\suc(q)=\nil$}
\end{array}
\right.
\]
\end{lemma}
Lemma~\ref{le:cost-update} implies that, after computing $\prd(q)$ and $\suc(q)$, 
we can update our data structure using $O(1)$ \emph{bulk updates} of the following form: 
\begin{quotation}
\noindent Given an interval $I$ of range values and an update value~$\Delta$, add $\Delta$ to the
cost of $\Gamma(P,p^*,\lambda_j)$ for all $\lambda_j\in I$.
\end{quotation}
We cannot afford to do this explicitly, so we implement bulk
updates by updating the auxiliary information stored in $O(\log n)$ nodes in~$\tree_{p^*}$, as follows.
\begin{enumerate}
\item Let $\lambda_{\min}$ and $\lambda_{\max}]$ be the two endpoints of the
      interval~$I$ (possibly $\lambda_{\max}]=\infty$). By searching with
      $\lambda_{\min}$ and $\lambda_{\max}]$ in $\tree_{p^*}$, identify a collection
      $\C(I)$ of $O(\log n)$ nodes in $\tree_{p^*}$ such that $\lambda_i\in I$
      if and only if the leaf storing~$\lambda_j$ is a descendant of a node in~$\C(I)$.
\item Add $\Delta$ to the correction values~$\Delta(v)$ of all  
      nodes~$v\in \C(I)$ and to the value~$\mc(v)$.
\item Update the values $\Delta(v)$, $\mc(v)$, and $\br(v)$
      of the $O(\log n)$ ancestors of the nodes in~$\C(v)$ in a bottom-up manner.
\end{enumerate}
Since algorithms for updating this type of auxiliary information are
rather standard we omit further details.

Besides updating the costs of the assignments $\Gamma(P,p^*,\lambda_j)$, 
we may also need to introduce another candidate range for~$p^*$.
In particular, we need to introduce the range~$|p^* q|$ if $|p^*q| > |p^*s|$.
To this end we need to compute the cost of the range assignment~$\Gamma(P,p^*,|p^* q|)$.
After computing~$\prd(q)$, $\suc(q)$, and the cost of $O(1)$ chains---this can
all be done in $O(\log n)$ time using the global tree~$\tree_P$---we can 
compute the cost of~$\Gamma(P,p^*,|p^* q|)$ in $O(1)$ time.
We then insert a leaf~$w$ for the range~$|p^* q|$ into~$\tree_{p^*}$
with $f(|p^* q|)$ equal to the just computed cost.
Finally, we update the values $\mc(v)$ and $\br(v)$ of the ancestors
$v$ of $w$ whose the current value of $\mc(v)$ is larger than $f(|p^* q|)$). 
(Re-balancing $\tree_{p^*}$, when necessary, can be done in a standard manner~\cite[Chapter~15]{intro-to-alg}.)  

The following lemma summarizes the discussion above.
\begin{lemma}\label{le:update-tree}
$\tree_{p^*}$ can be updated in $O(\log n)$ time per insertion and deletion.
\end{lemma}

\mypara{Putting it all together.}
To summarize, upon the insertion of a new point~$q$ into~$P$, we first update each 
tree~$\tree_{p^*}$, as described above. This takes  in $O(\log n)$ time per tree,
so $O(n\log n)$ time in total. Then we update the global tree~$\tree_P$ in $O(\log n)$ time. 
Finally, we create a tree $\tree_{q}$ with $q$ being the root-crossing point.
This can be done in $O(n\log n)$ time, by inserting the points
from $P$ one by one as described above. Thus inserting a new point~$q$
can be done in $O(n\log n)$ time in total, after which we know the cost of the
optimal solution for~$P\cup \{q\}$. Deletions can be handled in a similar manner,
so we obtain the following theorem.
\broadcastupdate*
\section{Missing proofs for Section~\ref{sec:sas}}
\label{appendix:omitted-details}

\subsection{Proof of Lemma~\ref{le:sas-stability}}
\label{app:proof-of-sas-stability}
\sasstability*
\begin{proof}
The range of a point~$p \in P\cup \{q\}$ can increase due to the insertion of~$q$ only if 
\begin{description}
\item[(i)] $p=q$ and $\rnew(q)>0$, or
\item[(ii)] $p$ is a zero-range point in $\rho_k(P)$, or 
\item[(iii)] $p$ is the root-crossing point in $\rho_k(P\cup\{q\})$, or
\item[(iv)] the standard range of $p$ increases due to the insertion of~$q$, or
\item[(v)] $p=s$ and, out of the two standard ranges it has, $s$ gets assigned a larger
    one in $\rk(P\cup\{q\})$ than in $\rk(P)$.
\end{description}
Recall that we defined $\rk$ such that the number of zero-range points is at most~$k$.
Furthermore, at most one standard range can increase due to the insertion of~$q$,
namely, the standard range of a point that is extreme in $P$ but not in $P\cup\{q\}$.
When this happens, however, $q$ is extreme in $P\cup\{q\}$ and so $\rnew(q)=0$; this implies that cases (i) and (iv) cannot both happen.
Hence, $\left| \{ p\in P\cup\{q\} : \rnew(p) > \rold(p) \} \right| \leq k+3$.
 
The range of a point~$p$ can decrease only if
\begin{description}
\item[(i)] $p$ is a zero-range point in $\rho_k(P\cup\{q\})$, or 
\item[(ii)] $p$ is the root-crossing point in $\rho_k(P)$, or
\item[(iii)] the standard range of $p$ decreases due to the insertion of~$q$, or
\item[(iv)] $p=s$ and, out of the two standard ranges it has, $p$ gets a assigned a smaller
one in $\rk(P\cup\{q\})$ than in $\rk(P)$.
\end{description}
Since the only point whose standard range decreases is the predecessor of~$q$
in $P$, we conclude that 
$\left| \{ p\in P\cup\{q\} : \rnew(p) < \rold(p) \} \right| \leq k+3$.
\end{proof}

\subsection{Proof of Theorem~\ref{thm:sas-alg}}
\label{app:sas-alg}
\sasalg*
\begin{proof}
Our SAS maintains the canonical range assignment~$\rk$ for 
$k=(2^{\alpha}/\eps)^{1/(\alpha-1)}=O((1/\eps)^{1/(\alpha-1)})$. 
We then have $\cost_{\alpha}(\rk(P))\leq (1+\eps) \cdot \ropt(P)$
by Lemma~\ref{le:sas-approx}. Furthermore, the number of modified ranges when~$P$
is updated is~$2k+6$ by Lemma~\ref{le:sas-stability}. To determine the assignment~$\rk$, 
we need to know an optimal assignment~$\ropt$ with the structure from Theorem~\ref{thm:structure-1D}.
Such an optimal assignment can be maintained in $O(n\log n)$ time per update, by
Theorem~\ref{thm:broadcast-update}. Once we have the new optimal assignment, the
new optimal assignment can be determined in~$O(n)$ time.
\end{proof}

\section{1-Stable, 2-Stable, and 3-Stable Algorithms in $\Reals^1$}
\label{app:123stable}
In Section~\ref{sec:sas} we presented a $(2k+6)$-stable algorithm whose approximation ratio
is~$1+2^{\alpha}/k^{\alpha-1}$, which provided us with a SAS. For small~$k$ the algorithm
is not very good: the most stable algorithm we can get is 6-stable,
by setting~$k=0$. A more careful analysis shows that the approximation ratio
of this 6-stable algorithm is~3, for $\alpha=2$. In this section we study more stable algorithms.
We first present a 1-stable $O(1)$-approximation algorithm; obviously,
this is the best we can do in terms of stability. This algorithm can only
handle insertions. (Deletions would be 3-stable.)
We then present a straightforward 2-stable 2-approximation
algorithm. Finally, we show that it is possible
to get an approximation ratio strictly below~2 using a 3-stable algorithm; see Table~\ref{ta:stableresults-app} for an overview.
%
\begin{table}[h]
\begin{center}
\begin{tabular}{c||c|l} 
$\ell$-stable algorithm & Approximation Ratio  & Remarks  \\ 
\hline
$\ell=1$  & $6+2\sqrt{5}\approx10.47$   & $\alpha=2$, insertions only \\
$\ell=2$  & 2                           & for any $\alpha > 1$  \\
$\ell=3$  & 1.97                        & $\alpha = 2$  \\ 
\end{tabular}
\caption{An overview of the approximation ratio of 1-stable, 2-stable and 3-stable algorithms}
\label{ta:stableresults-app}
\end{center}
\end{table}

\subsection{A 1-stable insertion-only algorithm}
\label{subsec:1-stable-alg}
We first describe our algorithm for the \emph{one-sided} version of the problem, where all
points in $P$ lie to the same side of the source. 
Let $P=\{s,p_1,\ldots,p_n\}$, where the points are numbered in order of increasing distance to
the source. It will be convenient to define $p_0:= s$. Our algorithm maintains a range
assignment~$\range$ that satisfies the following invariant.
\begin{itemize}
\item There is a path $\pi^*$ in $\cg(P)$ from $p_0$ to $p_n$ such that for each edge $(p_i,p_j)$
      on the path we have $\range(p_i) = |p_i p_j|$ and $i<j\leq i+4$. 
      For an edge $(p_i,p_j)$ on~$\pi^*$, we call the subsequence $p_i,\ldots,p_j$ a \emph{block},
      and we denote it by $B[i,j]$. 
\item A point $p_t$ in a block $B[i,j]$ is a zero-range point, unless $B[i,j]$ consists of
      five points (including $p_i$ and $p_j$) of which $p_t$ is the middle one.
      In the latter case $\range(p_t) = |p_t p_j|$.
\end{itemize}
Algorithm \textsc{1-Stable-Insert}, presented below, shows how to insert a point $q$ into $P$. 
Figure~\ref{fig:life-cycle} shows the life cycle of a block in the solution maintained
by the algorithm. 
\begin{figure}[b]
\begin{center}
\includegraphics{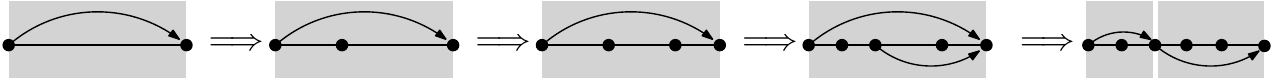}
\end{center}
\caption{Life cycle of a block. At the last step, the block is split into two smaller blocks,
which start in the middle of their life cycle: one block consists of three points, the other of four points.}
\label{fig:life-cycle}
\end{figure}
%
\begin{algorithm}[H] 
\caption{\textsc{1-Stable-Insert}$(P,q)$}
\begin{algorithmic}[1]
\State $\rhd$ By default $\range(q)=0$,  so we only set $\range(q)$ when it receives a non-zero range.
\If{$q$ is extreme}
\State Set $\range(\prd(q)) := |\prd(q) q|$, thus creating a new block.
\Else
\State Let $B[i,j]$ be the block containing~$q$, after the insertion of~$q$.
      \If{$B[i,j]$ consists of at most four points}
      \State Do nothing.
      \ElsIf{$B[i,j]$ consists of five points} 
      \State Set $\range(p_{\mathrm{mid}}) := |p_{\mathrm{mid}} p_j|$,
              where $p_{\mathrm{mid}}$ is the middle point from $B[i,j]$.
      \Else \hspace*{20mm} $\rhd$ $B[i,j]$ consists of six points 
      \State Let $p_{\mathrm{mid}}\not\in \{p_i,p_j\}$ be the point in $B[i,j]$ with non-zero range.
      \State Split the block $B[i,j]$ by decreasing the range of $p_i$ to $|p_i p_{\mathrm{mid}}|$.
      \EndIf
\EndIf
\end{algorithmic}
\end{algorithm}

It is readily verified that \textsc{1-Stable-Insert} maintains the invariant---hence, the
solution remains valid---and that it is 1-stable. We now analyze its approximation ratio.
\begin{lemma} \label{lem:1-stable-ratio}
Algorithm \textsc{1-Stable-Insert} maintains a $c'_{\alpha}$-approximation of an optimal solution
for the one-sided range-assignment problem in~$\Reals^1$, where the approximation ratio
$c'_{\alpha}$ depends on the distance-power gradient~$\alpha$. For $\alpha=2$ the
approximation ratio is $c'_2 = 3+\sqrt{5}$.
\end{lemma}
\begin{proof}
The unique optimal solution for the one-sided problem on the current point set $P=\{p_0,\ldots,p_n\}$
is the chain from $p_0$ to $p_n$, which has cost $\sum_{i=0}^{n-1} |p_i p_{i+1}|^{\alpha}$.
This implies that the approximation ratio of the current range assignment~$\range$
is bounded by the maximum, over all blocks $B[i,j]$ in the current assignment, of
the quantity
$
\sum_{t=i}^{j-1} \range(p_t)^{\alpha} \; \big/ \; \sum_{t=i}^{j-1} |p_t p_{t+1}|^{\alpha},
$
where the numerator gives the cost incurred by the algorithm on~$B[i,j]$
and the denominator gives the cost of the optimal solution on~$B[i,j]$.

To analyze this maximum, consider a block $B[i,j]$ and assume without loss
of generality that~$p_i=0$ and $p_j=1$. Clearly, the maximum approximation ratio that 
can be achieved is when $B[i,j]$ consists of five points; see the fourth block 
in Figure~\ref{fig:life-cycle}. Let $p_{(i+j)/2}$, the middle point in $B[i,j]$,
be located at position $x$, for some $0<x<1$. Then the cost of the algorithm incurred
on~$B[i,j]$ is $1 + (1-x)^{\alpha}$. The cost of the optimal solution on~$B[i,j]$
is minimized when the second point in the block is located at position~$x/2$ (this
is true because $\alpha>1$ and so for $0<y<x$ the function
$y^{\alpha} + (x-y)^{\alpha}$ is minimized at $y=x/2$) and,
similarly, when the fourth point is located at~$(x+1)/2$. Hence,
\[
\frac{\sum_{t=i}^{j-1} \range(p_t)^{\alpha}}{\sum_{t=i}^{j-1} |p_t p_{t+1}|^{\alpha}}
=
\frac{1 + (1-x)^{\alpha}}{2(x/2)^\alpha + 2((1-x)/2)^{\alpha}}
=
\frac{2^{\alpha-1} \cdot (1 + (1-x)^{\alpha})}{x^{\alpha} + (1-x)^{\alpha}}.
\]
Thus the approximation ratio is 
$
c'_{\alpha} = \max_{0\leq x \leq 1} \frac{2^{\alpha-1} \cdot (1 + (1-x)^{\alpha})}{x^{\alpha} + (1-x)^{\alpha}}.
$
For $\alpha=2$ this is maximized at $x=(3-\sqrt{5})/2$, giving an approximation
ratio~$c'_2 = 3+\sqrt{5}$.
\end{proof}
The approximation ratio of Algorithm \textsc{1-Stable-Insert} for the one-sided 
range assignment problem in $\Reals^1$ is actually tight for $\alpha =2$,
as the next observation shows.
\begin{observation}
For the one-sided range-assignment problem, the approximation ratio of $3+\sqrt{5}$ is tight for Algorithm~\textsc{1-Stable-Insert}.
\end{observation}
\begin{proof}
Let $c=(3-\sqrt{5})/2$, and consider the instance 
$P=\{s, p_1, p_2, p_3, p_4\}$ with $s=0,  p_1 = 1, p_2=\frac{c}{2}, p_3=c,p_4=\frac{(1+c)}{2}$. 
The insertion order of the points is $p_1, p_2, p_3, p_4$. Clearly, by setting 
$\rho(s)=\rho(p_2)=\frac{c}{2}$, and $ \rho(p_3) = \rho(p_4) = \frac{1-c}{2}$, 
an optimum solution with $\cost_2(\rho(P))= \frac{5-2\sqrt{5}}{2}$ is found.
Algorithm~\textsc{1-Stable-Insert} sets $\range(s):=1$ and $\range(p_3):=1-c$, and all other ranges to~0. 
The resulting cost equals $\frac{5-\sqrt{5}}{2}$, and the ratio follows.
\end{proof}

To handle the case with points to both sides
of~$s$, we proceed as follows. Let $P=L\cup \{s\}\cup R$, where $L$ and $R$ contain the
points to the left and to the right of $s$, respectively. We simply run the above algorithm
separately on $L\cup\{s\}$ and $\{s\} \cup R$. This way the points in~$R\cup L$ are assigned
one range, while~$s$ gets assigned two ranges; the actual range of~$s$ is the largest
of these two ranges. 
\begin{restatable}{theorem}{onestable}
\label{thm:1-stable}
There exists a 1-stable $c_{\alpha}$-approximation algorithm for the broadcast 
range-assignment problem in~$\Reals^1$, where the approximation ratio
$c_{\alpha}$ depends on the distance-power gradient~$\alpha$. For $\alpha=2$ the
approximation ratio is $c_2 = 2(3+\sqrt{5})\approx 10.47$.
\end{restatable}
\begin{proof}
Recall that our algorithm simply runs the one-sided algorithm
separately on $L\cup\{s\}$ and $\{s\} \cup R$,  where the actual range of~$s$
is defined to be the largest of the two ranges it receives.

To analyze the approximation ratio of this algorithm we use that for any~$\alpha>1$ we have
$\opt(L\cup\{s\}\cup R) \geq \max (\opt(L\cup\{s\}),\opt(\{s\}\cup R))$, where
$\opt(\cdot)$ denotes the cost of an optimal range assignment~\cite{DBLP:conf/isaac/Berg0U20}. 
Hence, the cost of the
range assignment~$\range$ that we maintain is
\[
\begin{array}{lll}
\costa(\range(L\cup\{s\}\cup R)) 
    & \leq & \costa(\range(L\cup\{s\})) + \costa(\range(\{s\}\cup R))  \\
    & \leq & c'_{\alpha} \cdot (\opt(L\cup\{s\}) + \opt(\{s\}\cup R)) \\
    & \leq & 2 c'_{\alpha} \cdot \max ( \opt(L\cup\{s\}) , \opt(\{s\}\cup R))
    \\
    & \leq & 2 c'_{\alpha} \cdot \opt(L\cup\{s\}\cup R). \\
\end{array}
\]
Lemma~\ref{lem:1-stable-ratio} thus implies the theorem.
\end{proof}

We note that a 1-stable algorithm \alg that handles deletions cannot have a bounded 
approximation ratio, as we show next for $\alpha=2$. Suppose for a contradiction that 
\alg has approximation ratio~$c$, where we assume for simplicity that $c$ is an integer.
Let $P := \{s,r_1,\ldots,r_{c+1}\}$ where $s=0$ and $r_i = i/(c+1)$.
Then $\opt=(c+1)\cdot (1/(c+1))^2=1/(c+1)$, so \alg cannot give the source a range of~1.
But if we then delete all non-zero points in $P\setminus\{s\}$, the algorithm is stuck:
the deletion of a non-zero point already causes a modification, so the algorithm
is not allowed to increase any range; hence, the solution is invalid
after all non-zero-range points from $P\setminus\{s\}$ have been deleted.
One may argue that it is unfair that the algorithm has to pay for the deletion of
a non-zero point. We defined it like this to keep the definition symmetric
for deletions and insertions, where the algorithm also has to pay for assigning
a non-zero range to a new point. It is then actually surprising that for insertions
it is possible to obtain a 1-stable algorithm with $O(1)$-approximation ratio.

\paragraph*{A lower bound for 1-stable algorithms}
The next theorem shows that any $1$-stable algorithm in~$\Reals^1$ has an approximation ratio greater than $2.61$ for $\alpha=2$. 
\begin{theorem} \label{thm:1-stable lower}
Any 1-stable algorithm for the dynamic broadcast range-assignment problem in $\Reals^1$ 
has an approximation ratio greater than or equal to $\frac{1}{2}\cdot(3+\sqrt{5})\approx 2.61$ 
for $\alpha=2$,
and any 1-stable algorithm has an approximation ratio greater than $2$ for $\alpha>2$. 
\end{theorem}

\begin{proof}
Let \alg be a $1$-stable algorithm, and let $\ralg$ be the range assignment 
it maintains. 

Consider the point set $P := \{s,r_1,r_2,p\}$, where $s=0$, and $r_i = x_i$ for $i=1,2$, 
and $p=1$. Assume $0<x_1<1$ and let $x_2=x_1+\frac{(1-x_1)}{2}$. Also assume after the source, 
the point $p$ arrives first, then the point $r_1$ and finally $r_2$ arrives. 
Let $P' := \{s,r_1,p\}$. Trivially, after the arrival of the point $p$, we must
have $\ralg\geq 1$ in order to have a feasible solution.
After the arrival of $r_1$, \alg is forced to keep $\ralg(s)\geq 1$ since \alg is 1-stable.

We consider two cases. 
\\[2mm]
\emph{Case I: After the arrival of $r_1$, \alg gives a range of at least $1-x_1$ to $r_1$} \\ 
In this case \alg cannot decrease any range. So,
      \[
      \cost_{\alpha}(\rho_{alg}(P')) \geq 1+(1-x_1)^\alpha
      \]
      An optimal solution on $P'$ has cost~$x_1^\alpha+(1-x_1)^\alpha$, 
      and so the approximation ratio of \alg in Case~~I is at least $\frac{ 1+(1-x_1)^\alpha}{x_1^\alpha+(1-x_1)^\alpha}$.
      \\[2mm]
\emph{Case II:  After the arrival of $r_1$, \alg gives a range less than $1-x_1$ to $r_1$ } \\
Now the point $r_2$ arrives. Since \alg is 1-stable it cannot decrease the range of the source, So,
\[
      \cost_{\alpha}(\rho_{alg}(P)) \geq 1
\]
An optimal solution on $P$ has cost~$x_1^\alpha+2\cdot (\frac{1-x_1}{2})^\alpha$, 
and so the approximation ratio of \alg in Case~~II is at least $\frac{ 1}{x_1^\alpha+2\cdot (\frac{1-x_1}{2})^\alpha}$.
\medskip

We conclude that the approximation ratio of any $1$-stable algorithm is greater than or equal to at least
\[
\min(\frac{ 1+(1-x_1)^\alpha}{x_1^\alpha+(1-x_1)^\alpha},\frac{ 1}{x_1^\alpha+2\cdot (\frac{1-x_1}{2})^\alpha}).
\]
For $\alpha=2$ we see that by substituting $x_1=\frac{3}{2}-\frac{\sqrt{5}}{2}$ we get the
approximation ratio is at least $\frac{1}{2}\cdot(3+\sqrt{5})$. Moreover, for any $\alpha>2$
we get an approximation ratio greater then~2, for instance 
by substituting $x_1=\frac{1}{2}$.
\end{proof}

\subsection{A 2-stable algorithm}
\label{subsec:2-stable}
Obtaining a 2-stable 2-approximation algorithm is straightforward: simply give every point in~$P$ 
its standard range, where the source~$s$ receives the largest of its (at most) two standard ranges. 
This induces a broadcast tree consisting of (at most)
two chains: a chain from~$s$ to the rightmost point and a chain from~$s$ to the leftmost point.
This algorithm is 2-stable: if we insert an extreme point then we increase the range of
at most one point, and if we insert a non-extreme point $q$ we increase the range of $q$
and decrease the range of its predecessor. (Deletions are symmetrical.) 
We call this algorithm the \emph{standard-range algorithm}.
It is easy show
that the standard-range algorithm gives a 2-approximation~\cite{DBLP:conf/isaac/Berg0U20}.
\begin{observation} \label{obs:MST}
The standard-range algorithm is a 2-stable 2-approximation algorithm for the dynamic broadcast range-assignment
problem in $\Reals^1$, for any power-distance gradient~$\alpha>1$. Moreover, the
approximation ratio of $2$ is tight for this algorithm.
\end{observation}
\begin{proof}
The fact that the approximation ratio is at most~2 was observed in \cite{DBLP:conf/isaac/Berg0U20}.
For completeness, we give an instance showing this bound is tight.

Define $P := P(\eps) \cup \{s\}$, where $s=0$ is the source, and $P(\eps) := \{p_1,p_2,p_3,p_4\}$,
where $p_1=\eps$, $p_2=-\eps$, $p_3=1$, and $p_4=-1$
for some small $\eps > 0$. The insertion order of the points in~$P(\eps)$ is $p_1,p_2,p_3,p_4$.
Clearly, by setting $\rho(s)=1$ and $\rho(p_i)=0$ for $i=1,\ldots,4$, we obtain an optimal solution 
with $\cost_{\alpha}(\rho(P))=1$. However, the standard-range algorithm 
will set $\rho(s) = \eps$, $\rho(p_1)=\rho(p_2)=1-\eps$, 
and $\rho(p_3)=\rho(p_4)=0$, leading to a solution with cost $f(\eps) := 2(1-\eps)^{\alpha}+\eps^{\alpha}$. 
Since $\lim_{\eps\rightarrow 0} f(\eps) = 2$,
this proves we have tightness for any $\alpha > 1$.
\end{proof}

\subsection{A 3-stable algorithm for $\alpha=2$ with approximation ratio less than~2}
\label{subsec:3-stable-lb}
Given the simplicity of our 2-stable 2-approximation algorithm, it is
surprisingly difficult to obtain an approximation ratio strictly smaller than~2.
In fact, we have not been able to do this with a 2-stable algorithm. 
Below we show this is possible with a 3-stable algorithm,
at least for the case $\alpha=2$, which we assume from now on.

Recall that for any set~$P$ with points on both sides of the source point~$s$, 
there is an optimal range assignment inducing
a broadcast tree with a single root-crossing point; see Figure~\ref{fig:optimal-structure}.
Unfortunately the root-crossing point may change when $P$ is updated.
This may cause many changes if we maintain a solution with a good 
approximation ratio and the same root-crossing point as the optimal solution. 
We therefore restrict ourselves to \emph{source-based range assignments},
where~$s$ is the root-crossing point. 
The main question is then how large the range of~$s$ should be, 
and which points within range of~$s$ should be zero-range points.
\medskip

We now define our source-based range assignment, which we denote by $\rsb$, more precisely.
It will be uniquely defined by the set $P$; it does not depend on the order in which points 
have been inserted or deleted.
Let $\delta$ be a parameter with $1/2 < \delta < 1$; later we will choose $\delta$ such
that the approximation ratio of our algorithm is optimized.
We call a point $p\in P\setminus\{s\}$ \emph{expensive}
if $\suc(p)\neq\nil$ and $|p\,\suc(p)| > \delta\cdot |s\,\suc(p)|$, and
we call it \emph{cheap} otherwise. 
The source~$s$ is defined to be always expensive. (This is consistent in the sense that for $p=s$
the condition $|p\,\suc(p)| > \delta\cdot |s\,\suc(p)|$ holds for both successors,
since $\delta<1$.)  We denote the set of all expensive points in $P$ by $\Pexp$
and the set of all cheap points by $\Pcheap$.
Define
$
\dmax := \max \{ |s\,\suc(p)| : p\in \Pexp \},
$
that is, $\dmax$ is the maximum distance from $s$ to the successor of any expensive point.
We say that a point $p\in\Pexp$ is \emph{crucial} if
$|s\,\suc(p)| = \dmax$. 
Typically there is a single crucial point, but 
there can also be two:  one on the left and one on the right of~$s$.
Our source-based range assignment~$\rsb$ is now defined as follows.
\begin{itemize}
\item $\rsb(s):= \dmax$,
\item $\rsb(p):= 0$ for all $p \in \Pexp \setminus \{s\}$, and
\item $\rsb(p):= \rst(p)$ for all $p \in \Pcheap$, where $\rst(p)$ denotes the standard range of a point.
\end{itemize}
The next lemma, whose proof is in the appendix, analyzes the stability of~$\rsb$.
The lemma implies that insertions are $(2,1)$-stable and deletions are $(1,2)$-stable.
\begin{restatable}{lemma}{threestable}
\label{le:(2,1)-stability}
Consider a point set $P$ and a point~$q\not\in P$. Let $\rold(p)$ be the range of a point~$p$
in $\rsb(P)$ and let $\rnew(p)$ be the range of~$p$ in $\rsb(P\cup \{q\})$. Then 
\[
\left| \{ p\in P\cup\{q\} : \rold(p) < \rnew(p) \} \right| \leq 2
\mbox{ and }
\left| \{ p\in P\cup\{q\} : \rold(p) > \rnew(p)) \right| \} \leq 1.
\]
\end{restatable}
\begin{proof}
Due to the insertion of $q$, five types of range modifications can happen. 
\begin{enumerate}[(i)]
\item The point $q$ may get a non-zero range because it is cheap and non-extreme.
\item A point~$p$ may move from $\Pcheap$ to $\Pexp$ and become a zero-range point.
      This can only happen when $p=\prd(q)$ and $p$ was extreme before the insertion of~$q$.
      Hence, $q$ will be extreme after its insertion, so this
      cannot occur together with type~(i).
\item A point~$p\in \Pcheap$ may get a smaller range because its standard range decreases. 
      This can only happen when $p=\prd(q)$, and so it cannot happen together with type~(ii).
\item A point $p$ may move from $\Pexp$ to $\Pcheap$ and get a non-zero range.
      Again, this can only happen when $p=\prd(q)$, so this cannot happen together with 
      types~(ii) or~(iii). 
\item The source~$s$ may get a different range because $\dmax$ changes. 
      If $\dmax$ decreases, then $\prd(q)$ must have been crucial, and so this cannot occur
      together with types~(ii) or~(iii).
      If $\dmax$ increases, then a type~(ii) modification must have occurred,
      which means that types~(i),~(iii), and~(iv) did not occur.
\end{enumerate}
Overall, we have at most one range increase of type~(i), at most one range change
from any of the types (ii), (iii), (iv), and at most one change of type~(v). 
There can be at most one decrease among these three changes, because
if type~(v) is a decrease then types~(ii) and (iii) did not occur.
Finally, there can be at most two increases, because if type~(v) is an increase
then types~(i),~(iii), and~(iv) did not occur.
\end{proof}
From now on we assume without loss of generality that the source~$s$ is located at $x=0$.
We will need the following lemma before we can proceed to prove the performance guarantee.
\begin{lemma}\label{Le_1}
Let $I\subset \Reals^1$ be an interval of length $\Delta_1$ at distance $\Delta_2$ from the source,
that is, $I=[\Delta_2,\Delta_2+\Delta_1]$ or $I=[-\Delta_1-\Delta_2,-\Delta_2]$, for some~$\Delta_2>0$. Let $\Pcheap(I)$ be the set
of all cheap points that are in $I$ and whose successor lies in~$I$ as well. 
Then $\sum_{p\in \Pcheap(I)} \rst(p)^2 \leq \delta (\Delta_1+\Delta_2)\Delta_1$.
\end{lemma}
\begin{proof}
Let $p\in \Pcheap(I)$. If $p$ is extreme we can ignore $p$ because $\rst(p)=0$,
so assume $p$ is not extreme. Since $p$ is cheap we then have
\[
\rst(p) = |p\,\suc(p)| \leq \delta \cdot |s \suc(p)|.
\]
Since $\suc(p)\in I$ we have $|s\suc(p)|\leq \Delta_1+\Delta_2$, and so $\rst(p)\leq \delta(\Delta_2+\Delta_1)$. Hence,
\[
\sum_{p\in \Pcheap(I)} \rst(p)^2 \ \ \leq \ \  \delta(\Delta_2+\Delta_1) \cdot \sum_{p\in \Pcheap(I)} \rst(p) 
\ \ \leq \ \ \delta(\Delta_2+\Delta_1)\Delta_1.
\]
\end{proof}
%
We now prove the approximation ratio of our source-based range assignment.
\begin{lemma}
\label{le:(2,1)-ratio}
For any point set $P$ in $\Reals^1$ and any $1/2 < \delta < 1$ we have
\[
\cost_2(\rsb(P))  \leq  \cdelta \cdot \opt, \mbox{ where $\cdelta := \max \left( 1+\delta+\frac{(1+5\delta)(1-\delta)^2}{\delta^2}, \frac{1}{\delta^2}+\frac{1}{2} \right)$} 
\]
and $\opt=\cost_2(\ropt(P))$ is the cost of the optimal range assignment on $P$.
\end{lemma}
\begin{proof}
The worst approximation ratio is achieved by a set~$P$ with points to both sides of 
the source---indeed, if we have only points to the right of~$s$, say, then adding an 
additional point slightly to the left of $s$ will change neither the cost of an
optimal solution nor the cost of~$\rsb$. So from now on we assume that $P$ has 
points to both sides of ~$s$. In the following, with a slight abuse of notation,
we will for a point $p\in P$ use the notation $p$ for the ``object''~$p$
and to its value (that is, its $x$-coordinate, if we identify $\Reals^1$ with the $x$-axis).
For example, to indicate that a point $p$ lies to the left of
another point~$p'$ we may write $p<p'$. We will assume without loss of generality that~$s=0$. 

Let $\ropt(P)$ be an optimal range assignment that induces a broadcast tree~$\B$ 
with the structure of Theorem~\ref{thm:structure-1D}, and let $q$ denote the 
root-crossing point in~$\B$.
With a slight abuse of notation, let $p^*_i$ denote a crucial point in $P$,
and let $p^*_{i+1}=\suc(p^*_i)$. 
If $p^*_i=s$ then we define $p^*_{i+1}$ to be a successor of~$s$ at maximum distance from~$s$,
so that also in this case we have $\dmax = |s p^*_{i+1}|$.
Thus $\rsb(s)=\dmax = |sp^*_{i+1}|$.

\begin{figure}
\begin{center}
\includegraphics{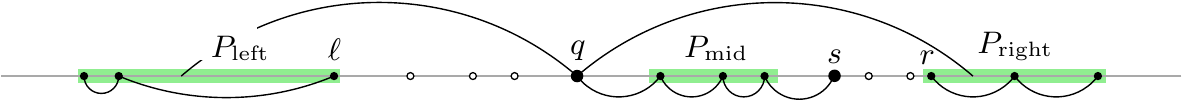}
\end{center}
\caption{The cost of $\ropt$ on $\Pleft\cup\Pmid\cup\Pright$ is at least the cost of
         $\rsb$ on this set. Note that $\ell$ and/or $r$ can lie exactly at the end
         of the range of~$q$, that is, $|q\ell|=\ropt(q)$ and/or $|qr|=\ropt(q)$---in fact, 
         one of this cases must happen.}
\label{fig:3-stable-proof-1}
\end{figure}
Let $\ell$ and $r$ be the leftmost and the rightmost points that are within
range of $q$ in the optimal solution, respectively.
Let $\Pleft := \{ p\in P: p\leq \ell\}$ be the set of all points to the left of $\ell$ 
plus $\ell$ itself, and let $\Pright := \{ p\in P: p\geq r\}$. Finally, let
$\Pmid$ be the set of points in between $s$ and $q$, excluding both~$s$ and~$q$;
see Figure~\ref{fig:3-stable-proof-1}.
We now define
\[
\Csb := \sum_{p\in \Pleft\cup\Pmid\cup\Pright} \rsb(p)^2
\mbox{\hspace*{5mm} and \hspace*{5mm}}
\Copt := \sum_{p\in \Pleft\cup\Pmid\cup\Pright} \ropt(p)^2
\]
as the costs incurred by $\rsb$ and $\ropt$ on the sets just defined.
Observe that $\Copt \geq \Csb$, because $\ropt(p)=\rst(p)$ 
for all $p\in \Pleft\cup\Pmid\cup\Pright$, and $\rsb(p)\leq\rst(p)$ for all $p\in P\setminus\{s\}$.

We now analyze the costs incurred by $\rsb$ and $\ropt$ on the remaining points.
We assume without loss of generality that $q\leq s$, and we let $x := |qs|$ 
denote the distance from $q$ to~$s$. Furthermore,
we define $z := \ropt(q)$. Note that $z\geq x$. 
We divide the analysis into several cases, depending on the relative position of~$s$, $q$, $p^*_i$, $p^*_{i+1}$.
\medskip

\noindent \emph{Case 1:  $p^*_i$ and $p^*_{i+1}$ lie to the right of~$s$ 
(possibly $p^*_i=s$) and inside the range of $q$ in the optimal solution.} \\[2mm]
\begin{figure}[h]
\begin{center}
\includegraphics{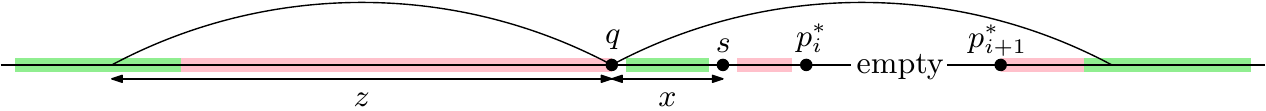}
\end{center}
\caption{Relative position of the points $s,q,p^*_i,p^*_{i+1}$ in Case 1. Costs of the points 
in the green regions are included in $\Csb$ and $\Copt$. The cost for the other regions is
analyzed in the text.}
\label{fig:case-1}
\end{figure}
See Figure~\ref{fig:case-1} for an illustration.
Since $p^*_i$ is crucial we have $|p^*_i p^*_{i+1}| > \delta\cdot |sp^*_{i+1}|$ and so
\[
|s p^*_i| = |s p^*_{i+1}| - |p^*_i p^*_{i+1}| 
          < \left( \frac{1}{\delta}-1\right)|p^*_i p^*_{i+1}|
          = \left( \frac{1-\delta}{\delta}\right)|p^*_i p^*_{i+1}|
          \leq \left( \frac{1-\delta}{\delta}\right) (z-x).
\]
We now bound the cost of the points $p\not\in \Pleft\cup\Pmid\cup\Pright$.
Theses are the points $s,q,p^*_i,p^*_{i+1}$ plus the points in the red regions
in Figure~\ref{fig:case-1}. Note that $\rsb(p^*_i)=0$.

By applying Lemma~$\ref{Le_1}$ with $\Delta_1\leq\left(\frac{1-\delta}{\delta}\right)(z-x)$
and $\Delta_2=0$, we see that the cost incurred by $\rsb$ due to the points 
strictly in between $s$ and $p^*_i$ is less than or equal to $\frac{(1-\delta)^2}{\delta}(z-x)^2$.
By applying Lemma~$\ref{Le_1}$ with $\Delta_1=z$ and $\Delta_2=x$, the cost incurred by $\rsb$ 
due to the points in the region left of and including~$q$ 
and within the range of $q$ is at most $\delta z(z+x)$.
Finally, the cost incurred by $\rsb$ due to the points to the right of and including $p^*_{i+1}$ 
and within the range of $q$ is at most $(z-x-|sp^*_{i+1}|)^2$.
Since $\rsb(s) = |sp^*_{i+1}|$ we obtain
\[
\cost_2(\rsb(P)) \ \leq \ |sp^*_{i+1}|^2 + \delta z(z+x) + \frac{(1-\delta)^2}{\delta}(z-x)^2+ (z-x-|sp^*_{i+1}|)^2+\Csb.
\]
Obviously, $\cost_2(\ropt(P))>z^2 + \Copt$.  Since $\Csb \leq \Copt$ and
$x \leq z$ we conclude
\[
\begin{array}{lll}
\frac{\cost_2(\rsb(P))}{\cost_2(\ropt(P))} 
  & \leq & \frac{(z-x-|sp^*_{i+1}|)^2+ |sp^*_{i+1}|^2 }{z^2} + \delta\frac{x}{z} + \frac{(1-\delta)^2}{\delta}  \\[2mm]
  & \leq & \frac{(z-x)^2}{z^2} + \delta\frac{x}{z} + \frac{(1-\delta)^2}{\delta}  \\[2mm]
  & \leq &  1+\delta+ \frac{(1-\delta)^2}{\delta} \\[2mm]
  & < &  \cdelta.
\end{array}
\]

\noindent \emph{Case 2: $p^*_i$ lies to the left of $q$ but within the range of $q$ and  
$p^*_{i+1}$ lies outside the range of $q$.} \\[2mm]
\begin{figure}[h]
\begin{center}
\includegraphics{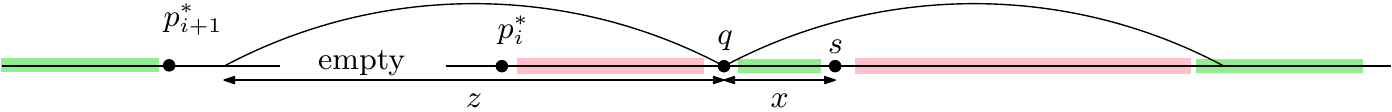}
\end{center}
\caption{Relative position of the points $s,q,p^*_i,p^*_{i+1}$ in Case 2.}
\label{fig:Case 2}
\end{figure}
As before, since $p^*_i$ is crucial we have $|p^*_i p^*_{i+1}| > \delta\cdot |sp^*_{i+1}|$ and so
\[
|s p^*_i| = |s p^*_{i+1}| - |p^*_i p^*_{i+1}| 
          < \left( \frac{1-\delta}{\delta}\right)|p^*_i p^*_{i+1}|.
\]
Applying Lemma~$\ref{Le_1}$ with $\Delta_1+\Delta_2=|sp^*_i|$ and 
$\Delta_1 = |qp^*_i| \leq z$,
we see that the cost incurred by $\rsb$ due to the points 
in the region between and including $q$ and $p^*_i$ is at most $(1-\delta)|p^*_ip^*_{i+1}|z$.
Applying Lemma~$\ref{Le_1}$ with $\Delta_1=z-x$ and $\Delta_2=0$, we see that
the cost incurred by $\rsb$ due to the points in the region to the right of~$s$ 
and within the range of $q$ is at most $\delta (z-x)^2$.
Furthermore, the cost incurred by $\rsb$ due to the range of $s$ is equal to $|sp^*_{i+1}|^2$.

Let $\Copt^* := \Copt-|p^*_ip^*_{i+1}|^2$.
Note that in Case~2, $p^*_i$ is the leftmost point in the range of $q$
(the point $\ell$ in Figure~\ref{fig:3-stable-proof-1}) and so it is included 
in $\Pleft$. Since $\rsb(p^*_i)=0$ this implies $\Copt^* \geq \Csb$.
Hence,
\[
\cost_2(\ropt(P))>z^2 +|p^*_ip^*_{i+1}|^2+  \Copt^*.
\]
Moreover,
\[
\cost_2(\rsb(P))= |sp^*_{i+1}|^2 +(1-\delta)|p^*_ip^*_{i+1}|z+\delta (z-x)^2+\Csb.
\]
Since $\Csb \leq \Copt^*$ we thus get
\[
\frac{\cost_2(\rsb(P))}{\cost_2(\ropt(P))}
\leq \frac{|sp^*_{i+1}|^2 +(1-\delta)|p^*_ip^*_{i+1}|z+\delta (z-x)^2}{z^2+|p^*_ip^*_{i+1}|^2}.
\]
Now if $|p^*_ip^*_{i+1}|\geq z$, and using that $x \leq z$, we find:
\[
\begin{array}{lll}

\frac{\cost_2(\rsb(P))}{\cost_2(\ropt(P))}  
 & \leq & \frac{|sp^*_{i+1}|^2}{|p^*_ip^*_{i+1}|^2} 
      + \frac{(1-\delta)|p^*_ip^*_{i+1}|z}{z^2+|p^*_ip^*_{i+1}|^2}
      + \frac{\delta (z-x)^2}{2z^2} \\[2mm]
 & \leq & \frac{1}{\delta^2}+\frac{1-\delta}{2}+\frac{\delta}{2} \\[2mm]
 &  =  & \frac{1}{2} + \frac{1}{\delta^2}.
\end{array}
\]
On the other hand, if $|p^*_ip^*_{i+1}|<z$ we have $|sp^*_{i+1}|<|p^*_ip^*_{i+1}|/\delta < z/\delta$
and so we get
\[
\begin{array}{lll}
\frac{\cost_2(\rsb(P))}{\cost_2(\ropt(P))}
  & \leq & \frac{|sp^*_{i+1}|^2}{z^2+|p^*_ip^*_{i+1}|^2}+\frac{(1-\delta)|p^*_ip^*_{i+1}|z}{z^2+|p^*_ip^*_{i+1}|^2}+\frac{\delta (z-x)^2}{z^2} \\[2mm]
   & < & \frac{|p^*_ip^*_{i+1}|z}{(z^2+|p^*_ip^*_{i+1}|^2)\delta^2}+\frac{(1-\delta)|p^*_ip^*_{i+1}|z}{z^2+|p^*_ip^*_{i+1}|^2}+\frac{\delta (z-x)^2}{z^2} \\[2mm]
   & \leq & \frac{1}{2\delta^2}+\frac{1-\delta}{2}+\delta.
\end{array}
\]
So we have,
\[\frac{\cost_2(\rsb(P))}{\cost_2(\ropt(P))}
    \leq \max \Big( \;\frac{1}{2} + \frac{1}{\delta^2}, \; \frac{1}{2\delta^2}+\frac{1-\delta}{2}+\delta \Big)
     =\frac{1}{\delta^2}+\frac{1}{2}
     \leq \cdelta.
\]

\noindent \emph{Case 3: $p^*_i,p^*_{i+1}$ lie to the left of $q$ and inside the range of $q$.} \\[2mm]
\begin{figure}[h]
\begin{center}
\includegraphics{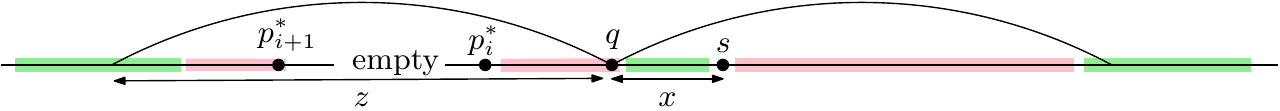}
\end{center}
\caption{Relative position of the points $s,q,p^*_i,p^*_{i+1}$ in Case 3. Costs of the points 
in the green regions are included in $\Csb$ and $\Copt$. The cost for the other regions is
analyzed in the text.}
\label{fig:Case 3}
\end{figure}
As before, since $p^*_i$ is crucial we have 
$|s p^*_i| < \left( \frac{1-\delta}{\delta}\right)|p^*_i p^*_{i+1}|$, which in Case~3 
implies $x \leq \frac{(1-\delta)}{\delta}z$ and $|qp^*_i| \leq \frac{(1-\delta)}{\delta}z$.
Applying Lemma~\ref{Le_1} with $\Delta_1=|qp_i^*|$ and $\Delta_2=x$, we see that
the cost incurred by $\rsb$ due to the points in the region 
between including $q$ and $p^*_i$ is at most 
$\delta \cdot z(x+\frac{(1-\delta)}{\delta}z) \cdot \frac{(1-\delta)}{\delta}$.
Applying Lemma~$\ref{Le_1}$ with $\Delta_1=z-x$ and $\Delta_2=0$, 
we see that the cost incurred by $\rsb$ due to the 
points in the region right of $s$ and within the range of $q$ is 
at most $\delta (z-x)^2$.
Finally, the cost incurred by $\rsb$ due to the points in the region left of $p^*_{i+1}$ 
and within the range of $q$ is at most $(z+x-|sp^*_{i+1}|)^2$.
Hence,
\[
\cost_2(\ropt(P))>z^2 + \Copt
\]
and
\[
\cost_2(\rsb(P))\leq 
|sp^*_{i+1}|^2 + \delta (z-x)^2  +\delta \cdot \frac{(1-\delta)}{\delta}
  \cdot z(x+\frac{(1-\delta)}{\delta}z)+ (z+x-|sp^*_{i+1}|)^2+\Csb, 
\]
Since $\Csb \leq \Copt$ this implies 
\[
\begin{array}{lll}
\frac{\cost_2(\rsb(P))}{\cost_2(\ropt(P))}
    & \leq & \frac{|sp^*_{i+1}|^2 + (z+x-|sp^*_{i+1}|)^2 + \delta (z-x)^2  +\delta \frac{(1-\delta)}{\delta}z(x+\frac{(1-\delta)}{\delta}z)  }{z^2} \\[2mm]
    & \leq & \frac{(z+x)^2+\delta(z-x)^2+(1-\delta)xz+\frac{(1-\delta)^2}{\delta}z^2}{z^2}  \ \mbox{(note that $|sp^*_{i+1}|^2 + (z+x-|sp^*_{i+1}|)^2< (z+x)^2$)} \\[2mm]
    & = & \frac{(1+\delta+\frac{(1-\delta)^2}{\delta})z^2+(3-3\delta)zx+(1+\delta)x^2}{z^2}  \\[2mm]
    & =  & (1+\delta)\frac{x^2}{z^2}+(3-3\delta)\frac{x}{z}+1+\delta+\frac{(1-\delta)^2}{\delta}.
\end{array}
\]
Now define $f(y) := (1+\delta)y^2+(3-3\delta)y+1+\delta+\frac{(1-\delta)^2}{\delta}$,
then the last term is equal to $f(x/z)$. 
Observe that~$f$ is quadratic in $y$ and that 
$0\leq x/z \leq (1-\delta)/\delta$.
Hence,
\[
\begin{array}{lll}
\frac{\cost_2(\rsb(P))}{\cost_2(\ropt(P))} 
 & \leq & \max \left( f(0),f((1-\delta)/\delta) \right) \\[2mm]
 & = & \max \left( 1+\delta+\frac{(1-\delta)^2}{\delta},1+\delta+ (1+5\delta)\frac{(1-\delta)^2}{\delta^2} \right) \\[2mm]
 & = & 1+\delta+ (1+5\delta)\frac{(1-\delta)^2}{\delta^2} \\[2mm]
 & \leq & \cdelta.
 \end{array}
\]
\emph{Case 4: $p^*_i,p^*_{i+1}$ lie to the left of $q$ but outside the range of $q$;
    or $p^*_i,p^*_{i+1}$ lie in the region $[s,q]$.} \\[2mm]
\begin{figure}[h]
\begin{center}
\includegraphics{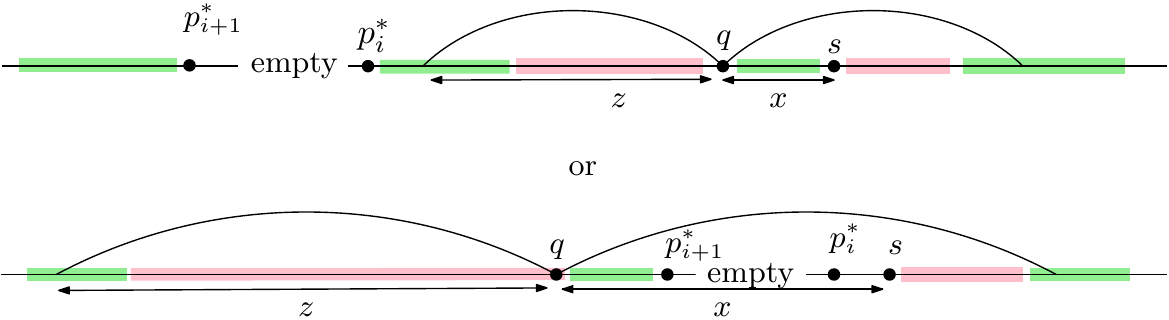}
\end{center}
\caption{The two options for relative position of the points $s,q,p^*_i$ and $p^*_{i+1}$ in Case 4.}
\label{fig:Case 4}
\end{figure}
Since $p_i^*$ is crucial we have $|sp^*_{i+1}| < |p^*_ip^*_{i+1}|/\delta$.
Clearly the cost incurred by $\rsb$ due to the points in the region to the left of $q$ 
and within the range of $q$ is at most~$z^2$.
Applying Lemma~$\ref{Le_1}$ with $\Delta_1=z-x$ and $\Delta_2=0$, the cost incurred by $\rsb$ 
due to the points in the region right of~$s$ and within the range of $q$ is at most~$\delta (z-x)^2$.
Finally, the cost incurred by $\rsb$ due to the range of~$s$ is~$|sp^*_{i+1}|^2$.
Define $\Copt^* := \Copt-|p^*_ip^*_{i+1}|^2$, and observe that $\Copt^* \geq \Csb$.
Then
\[
\cost_2(\ropt(P))>z^2 +|p^*_ip^*_{i+1}|^2+  \Copt^*
\]
and 
\[
\cost_2(\rsb(P))\leq |sp^*_{i+1}|^2 +z^2+\delta (z-x)^2+\Csb
\]
with $\Csb \leq \Copt^*$.
Since $x \leq z$ we thus obtain
\[
\frac{\cost_2(\rsb(P))}{\cost_2(\ropt(P))}
    \leq \frac{|sp^*_{i+1}|^2 +z^2+\delta (z-x)^2}{z^2+|p^*_ip^*_{i+1}|^2}
    \leq \frac{\frac{1}{\delta^2}|p^*_ip^*_{i+1}|^2| +(1+\delta)z^2}{z^2+|p^*_ip^*_{i+1}|^2}
    \leq \max\left(1+\delta,\frac{1}{\delta^2}\right) < \cdelta.
\]
\emph{Case 5: $p^*_i$ lies to the right of $s$ and within the range of $q$ and 
              $p^*_{i+1}$ lies to the right of~$s$ and outside the range of $q$;
              or $p^*_i,p^*_{i+1}$ lie to the right of~$s$ and outside the range of $q$.}\\[2mm]
\begin{figure}[h]
\begin{center}
\includegraphics{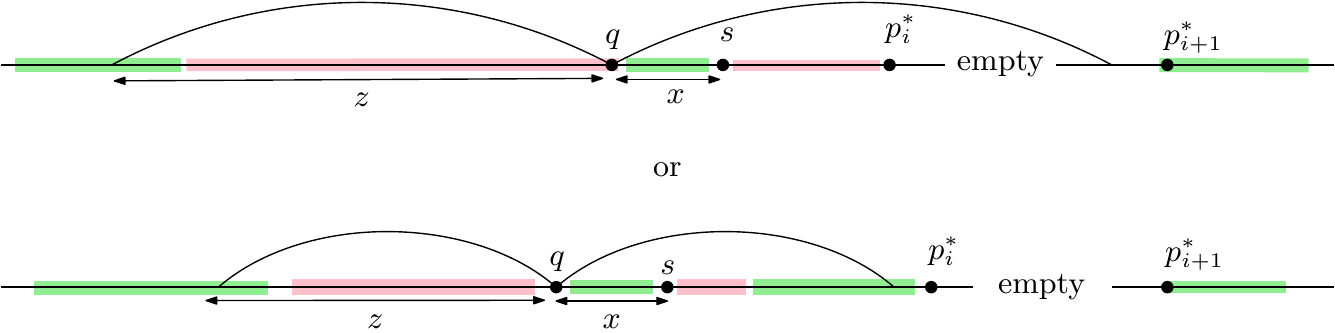}
\end{center}
\caption{The two options for the relative position of the points $s,q,p^*_i$ and $p^*_{i+1}$ in Case 5.}
\label{fig:Case 6}
\end{figure}
As before, we have $|sp^*_i| < \frac{(1-\delta)}{\delta}|p^*_ip^*_{i+1}|$
and $|sp^*_{i+1}| < \frac{1}{\delta}|p^*_ip^*_{i+1}|$.
Clearly the cost incurred by $\rsb$ due to the points in the region to the left 
of $q$ and within the range of $q$ is at most~$z^2$, and the cost incurred 
by $\rsb$ due to the points in the region right of~$s$ and within the range 
of $q$ at most $|sp^*_i|^2 < \frac{(1-\delta)^2}{\delta}|p^*_ip^*_{i+1}|^2$.
Finally, the cost incurred by $\rsb$ due to source~$s$ is 
$|sp^*_{i+1}|^2$ which is most $\frac{1}{\delta^2}|p^*_ip^*_{i+1}|^2$.
Define $\Copt^* := \Copt-|p^*_ip^*_{i+1}|^2$ and observe that $\Copt^* \geq \Csb$.
Then 
\[
\cost_2(\ropt(P)) > z^2 + |p^*_ip^*_{i+1}|^2 +  \Copt^*
\]
and 
\[
\cost_2(\rsb(P))\leq |sp^*_{i+1}|^2 +z^2+\frac{(1-\delta)^2}{\delta}|p^*_ip^*_{i+1}|^2+\Csb
\]
Since $\Csb \leq \Copt^*$ and $\delta < 1$ we conclude
\[
\begin{array}{lll} 
\frac{\cost_2(\rsb(P))}{\cost_2(\ropt(P))}
    & \leq & \frac{|sp^*_{i+1}|^2 +z^2+\frac{(1-\delta)^2}{\delta}|p^*_ip^*_{i+1}|^2}{z^2+|p^*_ip^*_{i+1}|^2} \\[2mm]
    & \leq & \frac{(\frac{1}{\delta^2}+\frac{(1-\delta)^2}{\delta})|p^*_ip^*_{i+1}|^2 +z^2}{z^2+|p^*_ip^*_{i+1}|^2} \\[2mm]
    & \leq & \frac{(1-\delta)^2}{\delta}+\frac{1}{\delta^2} \\[2mm]
    &  <   & \cdelta.
\end{array}
\]
We conclude that the approximation ratio is bounded by $\cdelta$ in all cases.
\end{proof}
We now want to choose $\delta$ so as to minimize $\cdelta=
\max \left( 1+\delta+\frac{(1+5\delta)(1-\delta)^2}{\delta^2}, \frac{1}{\delta^2}+\frac{1}{2} \right)$. The first term is minimized at the real root of the polynomial 
$6\delta^3-3\delta-2$, whose approximate value is $0.92711$; this gives a value
that is approximately~1.97. For this value of $\delta$ the first term dominates the second
one, leading to the following theorem.
\begin{theorem}
\label{thm:(2,1)}
There exists a 3-stable 1.97-approximation algorithm for the dynamic broadcast 
range-assignment problem in $\Reals^1$ for $\alpha=2$.  
\end{theorem}
\section{Missing proofs for Section~\ref{sec:S1}}
\label{app:S1}

\subsection{The structure of an optimal solution in~$\bbS^1$}
\label{app:structure-in-S1}
In this section we will prove Lemma~\ref{lem:S1-structure}, which states that for any instance 
in $\bbS^1$, there exists a point $r\in \bbS^1$ that is not within range of any point.
Hence, $\bbS^1$ can be cut at $r$ to obtain an equivalent instance in $\Reals^1$.

Without loss of generality we identify $\bbS^1$ with a circle of perimeter~1.
Let $\ropt$ be a fixed optimal range-assignment on $P$. 
We will need the following simple lemma.
\begin{lemma} \label{maximum range of a point}
If $|P|>2$ then $\ropt(p)<\frac{1}{2}$ for all $p\in P$.
\end{lemma}
\begin{proof}
Note that setting $\rho(s)=\frac{1}{2}$ and $\rho(p)=0$ for all $p\in P\setminus\{s\}$ gives
a feasible solution. Since $\rho(s)>0$ in any feasible solution, this means that $\ropt(p)<\frac{1}{2}$ for
all $p\neq s$. Hence, it suffices to 
show that $\ropt(s) < \frac{1}{2}$. If there is no point $p\in P$ which is diametrically 
opposite~$s$ then clearly $\ropt(s) < \frac{1}{2}$. Now suppose there is a point~$p\in P$ 
that lies diametrically opposite~$s$.
Let $q\in P\setminus \{s,p\}$ be a point that maximizes the distance from~$s$ 
among all points in $P\setminus \{s,p\}$. The point~$q$ exists since $|P|>2$. 
Note that $d(s,q)+d(q,p)=\frac{1}{2}$. Hence, setting $\rho(s)=d(s,q)$ and
$\rho(q)=d(q,p)$ (and keeping all other ranges zero) gives a solution
of cost $d(s,q)^{\alpha} + d(q,p)^{\alpha}$, which is less than $\left( \frac{1}{2}\right)^\alpha$ since $\alpha>1$.
Thus $\ropt(s)<\frac{1}{2}$, which finishes the proof.
\end{proof}
Before we proceed, we introduce some more notation.

For two points $p,q\in\bbS^1$, we let $[p,q]^{\cw}\subset \bbS^1$ denote the  
closed clockwise interval from $p$ to $q$. In other words, $[p,q]^{\cw}$ is the clockwise arc along~$\bbS^1$
from $p$ to $q$, including its endpoints. Furthermore, we define 
$(p,q)^{\ccw}$ to be the open clockwise interval from $p$ to~$q$.
The intervals $[p,q]^{\ccw}$ and $(p,q)^{\ccw}$
are defined similarly, but for the counterclockwise direction.
Now consider a directed edge $(p,q)$ in a communication graph~$\cg(P)$. We say that $(p,q)$ is a
\emph{clockwise edge} if $\rho(p)\geq \dcw(p,q)$, and we say that it is
a \emph{counterclockwise edge} if $\rho(p)\geq \dccw(p,q)$.
Observe that Lemma~\ref{maximum range of a point} implies that an edge cannot be both 
clockwise and counterclockwise in an optimal range assignment, assuming $|P|>2$.
%
Finally, we define  the \emph{covered region} of a subset $Q\subseteq P$ with respect to 
a range assignment $\rho$ to be the set of all points $r\in \bbS^1$ such that there 
exists a point $p\in Q$ such that $\rho(p)\geq d(p,r)$. We denote this region by $\cov(\rho,Q)$. 
Furthermore, the  \emph{counterclockwise covered region} of $Q$, denoted by $\cov_{\ccw}(\rho,Q)$, 
is the set of all points $r\in \bbS^1$ such that there exists a point $p\in Q$ such that 
$\rho(p)\geq \dccw(p,r)$. The \emph{clockwise covered region} of $Q$, denoted by $\cov_{\cw}(\rho,Q)$,
is defined similarly.

We can now state the main lemma of this section.
\Sstructure*
\begin{proof}
Let $\dhop(p,q)$ denote the hop distance from $p$ to $q$ in the communication graph~$\graph_{\ropt}(P)$.
Let~$\B$ broadcast tree rooted at $s$ in $\graph_{\ropt}(P)$ with the following properties. 
\begin{itemize}
    \item $\B$ is a shortest-path tree in terms of hop distance, that is, the hop-distance from $s$ to 
           any point $p$ in~$\B$ is equal to $\dhop(s,p)$. 
\item Among all such shortest-path trees,~$\B$ maximizes the number of clockwise edges.
\end{itemize}
For two points $p,q\in P$, let $\pi(p,q)$ denote the path from $p$ to $q$ in~$\B$,
and let $|\pi(p,q)|$ be its length, that is, the number of edges on the path.
Note that $|\pi(s,p)|=\dhop(s,p)$ for any $p\in P$.
Let $\pa(p)$ denote the parent of a point $p$ in~$\B$ and define
\[
S_{\cw}=\{p\in P\setminus\{s\}: \mbox{ $(\pa(p),p)$ is a clockwise edge} \}   
\]
and 
\[
S_{\ccw}=\{p\in  P\setminus\{s\}: \mbox{ $(\pa(p),p)$ is a counterclockwise edge} \}   
\]
Note that $S_{\cw} \cup S_{\ccw} = P\setminus \{s\}$. Now define 
\[
q_{\cw} = \mbox{the point from $S_{\cw}$ that maximizes $\dcw(s,p)$},
\]
where $q_{\ccw} = s$ if $S_{\cw}=\emptyset$. Similarly, define
\[
q_{\ccw} = \mbox{the point from $S_{\ccw}$ that maximizes $\dccw(s,p)$},
\]
where $q_{\ccw} = s$ if $S_{\ccw}=\emptyset$. 
Let $\anc(p)$ be the set of ancestors in $\B$ of a point $p\in P$, that is,
$\anc(p)$ contains the points of $\pi(s,p)$ excluding the point~$p$.
The following observation will be used repeatedly in the proof.
\begin{obsinproof} \label{obs:path-cover}
If $(\pa(p),p)$ is a clockwise edge, then $[s,p]^{\cw} \subset \cov(\ropt,\anc(p))$. 
Similarly, if  $(\pa(p),p)$ is a counterclockwise edge, 
then $[s,p]^{\ccw} \subset \cov(\ropt,\anc(p))$.
\end{obsinproof}
\begin{proofinproof}
Assume $(\pa(p),p)$ is a clockwise edge; the proof for when $(\pa(p),p)$ is a 
counterclockwise edge is similar. If $s\in[\pa(p),p)]^{\cw}$---this includes the
case where $\pa(p)=s$---then the statement obviously holds, so assume $\pa(p)\in [s,p]^{\cw}$.
Since $(\pa(p),p)$ is a clockwise edge, it then suffices to prove that $[s,\pa(p)]^{\cw} \subset \cov(\ropt,\anc(p))$.
Note that $\cov(\ropt,\anc(p))$ is connected, because the points in $\anc(p)$ form a path,
namely~$\pi(s,\pa(p))$. Since $\pi(s,p)$ is shortest path, $p\not\in \cov(\ropt,\anc(\pa(p))$,
which implies that $[s,\pa(p)]^{\cw} \subset \cov(\ropt,\anc(\pa(p)))\subset \cov(\ropt,\anc(p))$.
\end{proofinproof}
We now proceed to show that $q_{\ccw}$ must lie clockwise from $q_{\cw}$, as seen from~$s$,
that is, the situation shown in Fig.~\ref{fig:S1-structure}(i) cannot happen.
\begin{figure}
\begin{center}
\includegraphics{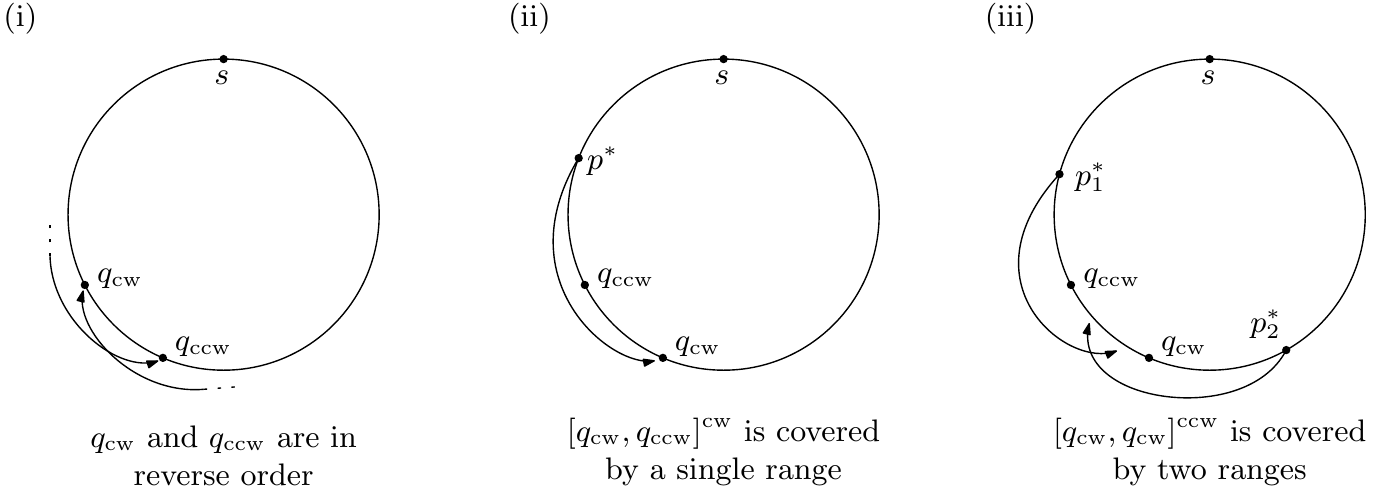}
\end{center}
\caption{Illustration for the proof of Lemma~\ref{lem:S1-structure}. Note that the point~$p^*$
in part~(ii) of the figure could also lie in~$[s,q_{\cw}]^{\cw}$. Similarly, in part~(iii)
the points~$p_1^*$ and $p_2^*$ could lie on ``the other side'' of~$s$.}
\label{fig:S1-structure}
\end{figure}
\begin{claiminproof}
$\dcw(s,q_{\cw})+\dccw(s,q_{\ccw}) < 1$.
\end{claiminproof}
\begin{proofinproof}
Note that $\dcw(s,q_{\cw})+\dccw(s,q_{\ccw}) \neq 1$, since otherwise
$q_{\cw}=q_{\ccw}$ which cannot happen since $S_{\cw}\cap S_{\ccw}=\emptyset$.

Now assume for a contradiction that  $\dcw(s,q_{\cw})+\dccw(s,q_{\ccw}) > 1$,
which means that $q_{\ccw} \in [s,q_{\cw}]^{\cw}$.
Since $q_{\cw}$ is reached from its parent by a clockwise edge,
this implies that $q_{\ccw} \in \cov(\ropt,\anc(q_{\cw}))$ by the observation above. 
Hence, $\dhop(s,q_{\cw}) \geq \dhop(s,q_{\ccw})$. 
An analogous argument shows that $\dhop(s,q_{\ccw}) \geq \dhop(s,q_{\cw})$.
Hence,  $\dhop(s,q_{\ccw}) = \dhop(s,q_{\cw})$. This implies that the edge $(\pa(q_{\cw}),q_{\cw})$ 
passes over~$q_{\ccw}$, otherwise some other edge of $\pi(s,q_{\cw})$ would pass over $q_{\ccw}$
and we would have $\dhop(s,q_{\ccw}) < \dhop(s,q_{\cw})$.
But then we also have a shortest path from $s$ to $q_{\ccw}$ 
whose last edge is a clockwise edge, contradicting the definition of~$\B$. 
\end{proofinproof}
So we can assume that $\dcw(s,q_{\cw})+\dccw(s,q_{\ccw}) < 1$ or, in
other words, that $q_{\ccw}$ lies clockwise from $q_{\cw}$, as seen from~$s$.
Clearly no point from $P$ lies in $(q_{\cw},q_{\ccw})^{\cw}$. If we have
$(q_{\cw},q_{\ccw})^{\cw}\not\subset \cov(\ropt, P)$ then we are done,
so assume for a contradiction that $(q_{\cw},q_{\ccw})^{\cw}\subset \cov(\ropt, P)$.
This can happen in three ways, each of which will lead to  a contradiction.
\medskip

\noindent \emph{Case~I: There exists a point $p^*\in \B$ such that $q_{\cw} \in \cov_{\ccw}(\ropt,\{p^*\})$.} 
\\[2mm]
See Fig.~\ref{fig:S1-structure}(ii) for an illustration of the situation.
If $p^*=s$ then $\dhop(s,q_{\cw})=1$. Since $q_{\cw} \in S_{\cw}$
this means that $q_{\cw}$ must also have an incoming clockwise edge from~$s$.
But then $\ropt(s)\geq \frac{1}{2}$, which
contradicts Lemma~\ref{maximum range of a point}. 
So $p^* \neq s$. 
Now note that $p^*$ must have an outgoing clockwise edge in~$\B$, else we can reduce the range of $p^*$ 
to $\dccw(p^*,q_{\ccw})$, which is smaller than $\dccw(p^*,q_{\cw})$, and still get a feasible solution. 
Observe that $p^* \notin \pi(s, q_{\cw})$; otherwise we must have $p^*=\pa(q_{\cw})$
(since $q_{\cw}$ lies in the range of~$p^*$) which contradicts that $q_{\cw}\in S_{\cw}$. 
So for any point from $P$ in the region $[s,q_{\cw}]^{\cw}$ there exists a path from~$s$ 
in the communication graph induced by $\rho_{opt}$ that does not use~$p^*$. 
We now have two subcases.

If $p^*\in [s,q_{\ccw}]^{\ccw}$ then clearly $p^* \in S_{\ccw}$
(otherwise the definition of $q_{\cw}$ is contradicted).
Hence, each point from $P$ in the region $[s,p^*]^{\ccw}$ has a path from $s$ 
that does not use~$p^*$. This implies that 
can reduce the range of $p^*$ to $\dccw(p^*,q_{\ccw})$ and still get a feasible solution. 

If $p^*\in [s,q_{\cw}]^{\cw}$ then obviously we can also reduce the range of $p^*$ to $\dccw(p^*,q_{\ccw})$ 
and still get a feasible solution. 

So both subcases lead to the desired contradiction. 
\medskip

\noindent \emph{Case~II: There exists a point $p^* \in \B$ such that $q_{\ccw} \in \cov_{\cw}(\ropt,\{p^*\})$}
\\[2mm]
In the proof of Case~I we never used that $\B$ maximizes the number of clockwise edges.
Hence, a symmetric argument shows that Case~II also leads to a contradiction.  
\medskip

\noindent \emph{Case~III: There are two points $p_1^*,p_2^* \in P$ such that $[q_{\cw},q_{\ccw}]^{\cw} \subseteq \cov_{\ccw}(\ropt,\{p_1^*\}) \cup \cov_{\cw}(\ropt,\{p_2^*\})$}.
\\[2mm]
See Fig.~\ref{fig:S1-structure}(iii) for an illustration of the situation.
We can assume that $q_{\cw} \notin \cov_{\ccw}(\ropt,\{p_1^*\})$ and $q_{\ccw} \notin  \cov_{\cw}(\ropt,\{p_2^*\})$,
otherwise we are in Case~I or Case~II.
Now either $p_2^* \notin \pi(s, p_1^*)$ or $p_1^* \notin \pi(s, p_2^*)$ or both. 
Without loss of generality, assume $p_2^* \notin \pi(s, p_1^*)$. 
Then $p_2^*\neq s$ and all points from $P$ in the region $[s,q_{\ccw}]^{\ccw}$ have a path from~$s$
in the communication graph~$\graph_{\ropt}(P)$ that does not use~$p_2^*$.
The point~$p_2^*$ must have an outgoing counterclockwise edge, else we can 
reduce the range of $p_2^*$ to $\dcw(p_2^*, q_{\cw})$ and still get a feasible solution.
We have two subcases.

If $p_2^* \in [s,q_{\ccw}]^{\ccw}$ then by reducing the range of $p_2^*$ to $\dcw(p_2^*, q_{\cw})$ 
we still get a feasible solution. 

If $p_2^* \in [s,q_{\cw}]^{\cw}$ then $p_2^*$ must be reached by a clockwise edge from its parent in~$\B$,
otherwise the definition of $q_{\ccw}$ would be contradicted. Hence, for each point from $P$ in the 
region $[s,p_2^*]^{\cw}$ there is a path from $s$ that does not use~$p_2^*$. 
So again we can reduce the range of $p_2^*$ to $\dcw(p_2^*, q_{\cw})$ we still get a feasible solution.

Thus both subcases lead to a contradiction.
\medskip

\noindent This finishes the proof of the lemma.
\end{proof}

\subsection{Missing proof for Section~\ref{subsec:S1-no-SAS}}
\label{app:no-sas-missing-proofs}

\costofpi*
\begin{proof}
By Observation~\ref{obs:S1-opt}, we have $\rho(p)^{\alpha} \leq \ca \cdot 2\delta^{\alpha}$ and, hence, $\rho(p) \leq (2\ca)^{1/\alpha} \cdot \delta < 3\delta$, for any point~$p$. 
Consider the interval $I=[y_1,y_2]^{\cw}$ where $\dcw(s,y_1)=3\delta$ and $\dccw(q,y_2)=3\delta$. 
All the points in $I\cap P$ are at distance more than $3\delta$ from $s$ or $q$ 
and hence $I\cap P \subseteq \cov(\ropt,P\setminus \{s,q\})$.
Let $p_i\in I\cap P$ be the point whose clockwise distance from~$s$ is minimum, 
and let  $p_j\in I\cap P$ be the point whose counterclockwise distance from $q$ is minimum. Then the cost of covering all the points in $I\cap P$ using the points in $P\setminus \{s,q\}$ is at least $\sum_{t=i}^{j-1} \dcw(p_t,p_{t+1})^{\alpha} -2^\alpha$, where the term $-2^{\alpha}$ is
because the covered region may leave one interval $[p_t,p_{t+1}]^{\cw}$ uncovered. 
Recall that the cost of assigning all the points in $P\setminus \{s,q\}$ a \CW-minimal range is $(2^\alpha+1)n$. 
Note that $i=O(\delta)$ since $\dcw(s,p_i) \leq 3\delta +2$ and $(2n+1)-j =O(\delta)$  since $\dcw(p_j,q) \leq 3\delta +2$.
Hence, 
\[
\sum_{i=1}^{2n+1}\rho(p_i)^\alpha 
\geq (2^{\alpha}+1)n - O(\delta)\cdot 2^{\alpha}
\geq (2^{\alpha}+1)n - o(n),
\]
since $\delta = ((2^{\alpha}+1)n)^{1/\alpha} = o(n)$.
\end{proof}

\nocworccwedge*
\begin{proof}
Suppose before insertion of $q$ the point~$p_{2n+1}$ has an incoming counterclockwise edge. 
The cheapest incoming counterclockwise edge would be from~$s$ and 
this is already too expensive. 
Indeed, if $\range(s)\geq 2x\delta$ then by Lemma~\ref{lem:cost-of-pi}
the total cost of the range assignment by $\alg$ is at least  
\begin{eqnarray*}
(2x\delta)^\alpha + (2^{\alpha}+1)n - o(n)
& = & \left( 2^{\alpha}\cdot \left( \frac{1}{4}+\left(\frac{1}{2}\right)^{\alpha+1} \right) + 1 \right) \cdot \delta^{\alpha} - o(n) \\
& = & \left(2^{\alpha-2} + \frac{3}{2} \right) \cdot \delta^{\alpha} - o(n) \\[1mm]
& = & \left(2^{\alpha-3} + \frac{3}{4} \right) \cdot 2 \delta^{\alpha} - o(n) \\[1mm]
& = & \left(1 + \left( 2^{\alpha-3} - \frac{1}{4} \right)  \right) \cdot 2 \delta^{\alpha} - o(n) \\[1mm]
& > & \left(1 + \frac{1}{2}\cdot \left( 2^{\alpha-3} - \frac{1}{4} \right)  \right) \cdot 2 \delta^{\alpha}  \hspace*{10mm} \mbox{for $n$ sufficiently large} \\[1mm]
& \geq & \ca \cdot \opt(P) \hspace*{12mm}  \mbox{by definition of $\ca$ and Observation~\ref{obs:S1-opt}}
\end{eqnarray*}
This contradicts the approximation ratio of \alg, proving the first part of the lemma.

Now suppose after the insertion of $q$ the point~$p_1$ has an incoming clockwise edge. The cheapest way
to achieve this is with $\rho(s)=\delta$, which is too expensive. Indeed,
by Lemma~\ref{lem:cost-of-pi} the total cost of the range assignment is then
at least 
\begin{eqnarray*}
\delta^\alpha + (2^{\alpha}+1)n - o(n) 
& = & \frac{2\delta^\alpha}{(2x^\alpha+1) \delta^\alpha} \cdot
      (2x^\alpha+1)\delta^\alpha - o(n) \\
& \geq & \left( 1+ \frac{1}{2}\cdot \left( \frac{2\delta^\alpha}{(2x^\alpha+1)\delta^\alpha}-1 \right) \right) \cdot  \opt(P\cup\{q\}) 
           \ \ \mbox{ for $n$ sufficiently large}  \\
& = & \left( 1+ \left( \frac{1}{2x^\alpha+1}-\frac{1}{2} \right) \right) \cdot  \opt(P\cup\{q\})  \\
& = & \left( 1+  \frac{2 - (2x^{\alpha}+1)}{2(2x^\alpha+1)}  \right) \cdot  \opt(P\cup\{q\})  \\
& = & \left( 1+  \frac{1 - \left( \frac{1}{2} + \frac{1}{2^{\alpha}}\right)}{2\left( \frac{1}{2} + \frac{1}{2^{\alpha}}+1\right)}  \right) \cdot  \opt(P\cup\{q\}) \ \ \mbox{ since $2x^{\alpha} = \frac{1}{2}+\frac{1}{2^{\alpha}}$} \\
& = & \left( 1+  \frac{2^{\alpha-1}-1}{3\cdot 2^{\alpha}+2}  \right) \cdot  \opt(P\cup\{q\}) \\       
& \geq & \ca \cdot \opt(P\cup\{q\}) \hspace*{10mm}  \mbox{by definition of $\ca$ and Observation~\ref{obs:S1-opt}}
\end{eqnarray*}
%
%
This contradicts the approximation ratio of \alg, proving the second part of the lemma.
\end{proof}

\paragraph*{The claim in the proof of Lemma~\ref{lem:many-minimal-ranges}}
\begin{claiminproof}
If $p_j$ is not assigned to a \CW-minimal edge then $\exc(p_j)\geq c'_{\alpha}$,
where $c'_{\alpha} = \min \left( 2^{\alpha}-1, \frac{3^{\alpha}-2^{\alpha}-1}{2}, \frac{ 4^{\alpha}-2^{\alpha}-2}{3}  \right)$.
\end{claiminproof}
\begin{proofinproof}
Consider a non-\CW-minimal edge~$(p_i,p_t)$. 
First suppose only a single point $p_j$ is assigned to~$(p_i,p_t)$. 
Then $t=i+1$ and $p_j=p_t$. Hence, $\range(p_i)\geq d(p_{j-1},p_j)+1$ 
because we assumed $\range(p_i)\in\{0,1,2\}\cup [3,\infty)$.
Thus when $|A(p_i,p_t)|=1$ then
\[
\exc(p_j) \geq (d(p_{j-1},p_j)+1)^{\alpha} - d(p_{j-1}p_j)^{\alpha}\geq 2^{\alpha}-1  \geq c'_{\alpha}.
\]
Now suppose $|A(p_i,p_t)|>1$. Let $z_1$ be the number of points $p_j\in A(p_i,p_t)$
with $d(p_{j-1},p_j)=1$, and let $z_2$ be the number of points $p_j\in A(p_i,p_t)$
with $d(p_{j-1},p_j)=2$. 
Since $|A(p_i,p_t)|>1$ we have $z_1\geq 1$ and $z_2\geq 1$ and $|z_1-z_2|\leq 1$.
When $|A(p_i,p_t)|=2$ then $z_1=z_2=1$, and we are distributing the cost of
an edge of length at least~3, minus the costs of edges of length~2~and~1,
over two points. Thus in this case we have
\[
\exc(p_j) \geq \frac{ 3^{\alpha}-2^{\alpha}-1}{2}.
\]
Similarly, when $|A(p_i,p_t)|=3$ then $z_1=2$ and $z_2=1$
(or vice versa, but that will only lead to a larger excess),
and we have
\[
\exc(p_j) \geq \frac{ 4^{\alpha}-2^{\alpha}-2}{3}.
\]
It remains to argue that we do not get a smaller excess when $|A(p_i,p_t)|\geq 4$.
To see this, we compare the excess we get when $(p_i,p_t)$ is an edge of $\pi$
with the excesses we would get when, instead of $(p_i,p_t)$, the edges $(p_i,p_{i+2})$ and $(p_{i+2},p_t)$
would be in~$\pi^*$.
Note that
\[
d(p_i,p_t)^\alpha = \Big( d(p_i,p_{i+2}) + d(p_{i+2},p_t) \Big) ^\alpha
                  > d(p_i,p_{i+2})^\alpha + d(p_{i+2},p_t)^\alpha
\]
since $\alpha>1$.  Hence,
\[
\begin{array}{lll}
\frac{ d(p_i,p_t)^\alpha - \sum_{\ell=i+1}^t d(p_{\ell-1},p_{\ell})^\alpha }{t-i} & > &
\frac{ \left( d(p_i,p_{i+2})^\alpha - \sum_{\ell=i+1}^{i+2} d(p_{\ell-1},p_{\ell})^\alpha \right) + 
\left( d(p_{i+2},p_t)^\alpha - \sum_{\ell=i+3}^t d(p_{\ell-1},p_{\ell})^\alpha \right) }{t-i}  \\[3mm]
& \geq &
\frac{ d(p_i,p_{i+2})^\alpha - \sum_{\ell=i+1}^{i+2} d(p_{\ell-1},p_{\ell})^\alpha }{2}+ 
\frac{ d(p_{i+2},p_t)^\alpha - \sum_{\ell=i+3}^t d(p_{\ell-1},p_{\ell})^\alpha }{t-i-2} 
\end{array}
\]
where the last inequality uses that
$\frac{a_1+a_2}{b_1+b_2} \geq \min \left( \frac{a_1}{b_1}, \frac{a_2}{b_2} \right)$
for any $a_1,a_2,b_1,b_2> 0$. Thus the excess we get for $(p_i,p_t)$ 
is at least the minimum of the excesses we would get for $(p_i,p_{i+2})$ and $(p_{i+3},p_t)$.
More generally, when $|A(p_i,p_t)|>4$ then we can compare the excess for
$(p_i,p_t)$ with the excesses we get when we would replace $(p_i,p_t)$
with a path of smaller edges, each being assigned two or three points.
The excess for $(p_i,p_{i+2})$ is at least the minimum of the excesses for these shorter edges. 
(Reducing to edges that are assigned a single point is not useful,
since these may be \CW-minimal and have zero excess.)
This finishes the proof of the claim.
\end{proofinproof}
\section{Missing details for Section~\ref{sec:higher-dim}}
\label{appendix:higher-dim}

\subsection{Proof of Theorem~\ref{thm:no-SAS-R2}}
\label{app:no-SAS-R2}

\noSASplane*
\begin{proof}
We use the same construction as in $\bbS^1$, where we embed the points on a square and
the distances used to define the instance are measured along the square; see Fig.~\ref{fig:no-sas-cirle-new}(ii). 
We now discuss the changes needed in the proof to deal with the fact that distances in~$\Reals^2$
between points from $P\cup\{q\}$ may be smaller than when measured along the square.
With a slight abuse of terminology, we will still refer to an edge $(p,p')$
that was clockwise in $\bbS^1$ as a clockwise edge, and similarly for counterclockwise edges.

Note that Observation~\ref{obs:S1-opt} still holds. 
Now consider Lemma~\ref{lem:cost-of-pi}.
The proof used that the points $p_i$ at distance more than $3\delta$ from $s$ or $q$ must be
covered by the ranges of the points $p_1,\ldots,p_{2n+1}$. We now restrict our attention to the points
that are also at distance more than $3\delta$ from a corner of the square. 
Each such point~$p_i$ must be covered by the range of some point~$p_j$ on the same edge of the square.
Hence, the distance in $\Reals^2$ of from $p_j$ to $p_i$ is the same as the 
distance in~$\bbS^1$, so we can use the same reasoning as before. Thus the exclusion of the 
points that are at distance at most $3\delta$ from a corner
of the square only influences the constant in the $o(n)$ term in the lemma. 
Hence, Lemma~\ref{lem:cost-of-pi} still holds.

The proof of Lemma~\ref{lem:no-cw-or-ccw-edge} still holds, since the cheapest
counterclockwise edge to $p_{2n+1}$ before the insertion of $q$ is still from $s$
(and the distance from $s$ to $p_{2n+1}$ did not change), and the cheapest
clockwise edge to $p_1$ after the insertion of $q$ is still from $s$
(and the distance from $s$ to $p_{1}$ did not change).

It remains to check Lemma~\ref{lem:many-minimal-ranges}. The proof still holds, except
that the claim that $\exc(p_j)\geq c'_{\alpha}$ may not be true for the given
value of $c'_{\alpha}$ when~$p_j$ is near a corner
of the square, because the distances between points on different
edges of the square do not correspond to the distances in $\bbS^1$.
To deal with this, we simply ignore the excess of any point within
distance $3\delta$ from a corner. This reduces the total excess by~$o(n)$. 
It is easily verified that this does not invalidate
the rest of the proof: we have to subtract $o(n)$ from the
formulae in Equality~(\ref{eq2}), but 
this is still larger than~$\ca\cdot\opt(P)$.

We conclude that all lemmas still hold, which proves Theorem~\ref{thm:no-SAS-R2}.
\end{proof}

\subsection{An $O(1)$-stable $O(1)$-approximation algorithm in $\Reals^2$}
\label{app:mst-in-R2}
Next we show that there is a relatively simple $O(1)$-stable $O(1)$-approximation 
algorithm for distance-power gradient $\alpha \geq 2$.
Our algorithm is based on a result by Amb\"uhl~\cite{DBLP:conf/icalp/Ambuhl05}, who showed that a minimum-spanning tree (MST)
on~$P$ can be used to get a constant-factor approximation to the 
broadcast range-assignment problem on~$P$. The key lemma underlying the result is the following.
\begin{lemma}[\cite{DBLP:conf/icalp/Ambuhl05}]
\label{lem:mst-cost}
Let $P$ be any point set in the plane. Let $T_P$ be an MST on $P$, and let $E(T_P)$ be 
the set of edges of~$T_P$. Then, for any distance-power gradient $\alpha\geq 2$, we have
$\sum_{e\in T_P} |e|^\alpha \leq 6 \cdot \opt$,
where $\opt = \costa(\ropt(P))$ is the cost of an optimal range assignment.
\end{lemma}
In the static problem this immediately gives a 6-approximation algorithm: 
turn the MST into a directed tree rooted at the source~$s$, and assign 
as a range to each point $p\in P$  the maximum length of any of its outgoing edges.
To apply this in the dynamic setting, we need the following lemma, which implies that 
for any point set $P$ and any additional point~$q$, any MST $T$ 
on $P$ can be converted to an MST $T'$ on $P\cup\{q\}$ that is very similar to $T$.
The result is folklore~\cite{DBLP:journals/siamcomp/SpiraP75}.
\begin{lemma}\label{le:mst-update}
Let $P$ be a set of points in a metric space $X$, and $P' := P \cup \{q\}$ 
for some point~$q\in X$. Let $T$ be any MST on~$P$.
Then there is an MST $T'$ on $P'$ 
such that all edges in $T'$ that are not incident to $q$ also occur in~$T$.
Conversely, let $T$ be any MST on a set $P\cup\{q\}$. Then there exists an MST $T'$
on $P$ such that all edges in $T$ that are not incident to $q$ also occur in~$T'$.
\end{lemma}
We use this lemma in combination with the following well-known lemma. 
\begin{lemma}\label{degree}
Let $T$ be an MST of a point set in $\Reals^d$. Then the maximum vertex degree of
$T$ is bounded by the Hadwiger number of the corresponding unit ball in $\Reals^d$.
In particular, the maximum vertex degree of an MST in $\Reals^2$ is at most $6$.
\end{lemma}
We can now prove the following theorem.
\begin{theorem} \label{2D approximate MST}
There is a 17-stable 12-approximation algorithm for the dynamic broadcast range-assignment
problem in $\Reals^2$, for power-distance gradient~$\alpha \geq 2$.
\end{theorem}
\begin{proof}
Our algorithm will maintain an MST $T$ on the current point set~$P$, using Lemma~\ref{le:mst-update}.
We set the range of each
point to be the maximum length of any of its incident edges. 
Clearly, this defines a feasible solution.
We denote the resulting range assignment by~$\rmst$.

We now analyze the stability of~$\rmst$. Consider the insertion of a point~$q$,
and let $T'$ be the new MST after the insertion has been handled.
Observe that, part from the point $q$ itself, only the ranges of the neighbors of
$q$ in $T'$ can increase. By Lemma~\ref{degree} we have $\degree(q)\leq 6$,
where $\deg(q)$ denotes the degree of $q$. 
Hence, the number of ranges that need to be increased is at most~$7$. 
Also observe that only the ranges of those points can decrease
that had an edge belonging to the edge set $T\setminus T'$ incident to it. 
Since $\degree(q)\leq 6$, and $T$ and $T'$ have $|P|-1$ and $|P|$ edges, respectively,
we have $|T\setminus T'|\leq 5$. Hence the ranges of at most ten points can decrease. 
So insertions are (7,10)-stable, and deletions are (10,7)-stable.

It remains to prove the approximation ratio. Using Lemma~\ref{lem:mst-cost}
and noting that every edge in $T$ is adjacent to at most two vertices, 
we conclude that
$\costa(\rmst(P)) \leq 2\cdot \sum_{e\in T} |e|^\alpha \leq 12\cdot \opt$
for any set $P$, thus finishing the proof.
\end{proof}
\section{Some additional previous work on stable algorithms}
\label{sec:stab-literature}

For an arbitrary dynamic optimization problem, 
stability of a solution depends on how to measure the difference between two solutions~\cite{wulms2020}. 
For instance, in our dynamic broadcast range-assignment problem, we focus on the number of points 
whose range changes after an insertion or a deletion. For many problems, there is a natural choice 
for the difference between two solutions. In fact, for some specific problems, trade-offs between 
the quality of a solution and its stability have been studied before, albeit with varying terminology. 
We now describe a number of such cases (without claiming to give a complete overview).

Lattanzi and Vassilvitskii~\cite{DBLP:conf/icml/LattanziV17} consider an online clustering problem: 
points arrive iteratively, and the goal is to maintain a set of $k$ centers (where $k$ is given) 
such that each point is assigned to a center while minimizing some given cost function; 
this setting encompasses the $k$-center, the $k$-median and the $k$-means problem. 
Lattanzi and Vassilvitskii explicitly focus on the trade-off between the cost of a solution and its 
stability---they use the term ``consistency''---which is defined as the sum, over all iterations, 
of the quantity $|\C_{i+1}\setminus \C_{i}|$, where $\C_i$ denotes the set of centers used in the
solution after the $i$-th iteration.
They show how to maintain a solution that is a constant-factor approximation while the 
consistency is $O(k~\mbox{log}~n)$, where $n$ stands for the number of arriving points. 
Fichtenberger~\etal~\cite{DBLP:conf/soda/FichtenbergerLN21} build 
on this, and show how to maintain a constant-factor approximation for the $k$-median problem 
by performing an optimal (up to polylogarithmic factors) number of center swaps.
Results of this nature for the uncapacitated facility-location problem 
can be found in Cohen-Addad~\etal~\cite{DBLP:conf/nips/Cohen-AddadHPSS19}.

In the online Steiner Tree problem, new points arrive iteratively, and one has to 
connect each new point to the current tree, thereby maintaining a tree on all points that have arrived. 
This problem was studied from a stability viewpoint by Imase and Waxman~\cite{DBLP:journals/siamdm/ImaseW91}. 
They show that a simple greedy algorithm that connects the new point to a closest previous point,
is $\log n$ competitive, where $n$ is the number of arriving points; 
they also show that it is possible to maintain a 2-competitive tree while making $O(n^{\frac32})$ swaps. 
Subsequent contributions (see \cite{DBLP:journals/siamcomp/MegowSVW16,DBLP:conf/soda/GuptaK14}) have improved those results; in particular, for the metric case, Gu~\etal~\cite{DBLP:journals/siamcomp/GuG016} 
provide an algorithm that, given any fixed $\eps>0$, maintains a tree that is $(1+\epsilon)$-competitive
by making at most $2n/\eps$ swaps over the course of $n$ arrivals. 
In our terminology (see Definition~\ref{def:SAS}) this would be a SAS for metric Steiner Tree,
albeit only in an amortized sense and for the insertion-only setting.

In dynamic scheduling problems, jobs arrive and disappear in an online fashion; 
the goal is to maintain an assignment of jobs to machines that is close to optimal. 
Clearly, the possibility of re-assigning jobs will help in maintaining high-quality solutions. 
The cost of re-assigning jobs can be measured by the number of jobs one is allowed 
to re-assign (the recourse model), or by the total size of the jobs that can be 
re-assigned (the migration model). When the objective is to minimize the makespan, 
the trade-off between the cost of re-assigning jobs and quality of the solution 
has been studied under the name \emph{robust PTAS}: a polynomial-time algorithm
that, for any given parameter~$\eps>0$, computes a $(1+\eps)$-approximation to the optimal
solution with re-assignment costs only depending on~$\eps$; 
we mention Skutella and Verschae~\cite{DBLP:conf/esa/SkutellaV10}, 
Sanders~\etal~\cite{DBLP:journals/mor/SandersSS09} for such results.
Similarly, in dynamic bin packing, there is a trade-off between maintaining 
a high-quality solution and the amount of repacking that is allowed.  
We mention work by Feldkord~\etal~\cite{DBLP:conf/icalp/FeldkordFGGKRW18}, 
Berndt~\etal~\cite{DBLP:journals/mp/BerndtJK20}, Epstein and Levin~\cite{DBLP:journals/mp/EpsteinL09}. 

For the dynamic matching problem, an impressive amount of work has been devoted 
trade-offs between the quality of a matching and the time needed to update it;
see Behnezhad~\etal~\cite{ DBLP:conf/soda/BehnezhadLM20} and the references contained therein. 
In the context of bipartite graphs, Bernstein~\etal~\cite{DBLP:journals/jacm/BernsteinHR19} consider the problem of maintaining a maximum matching, while minimizing the number of replacements, see also Gupta~\etal~\cite{DBLP:conf/soda/GuptaKS14}.


\end{document}